\documentclass[11pt]{article}
\usepackage[margin=1in]{geometry}
\usepackage{times,mathptmx}
\usepackage{graphicx,color,hyperref}
\usepackage{xspace,amsmath,amsfonts,amsthm,float,verbatim,enumitem}

\title{Annotations for Sparse Data Streams}

\author{
Amit Chakrabarti%
\thanks{Department of Computer Science, Dartmouth College. Supported in part by NSF grant CCF-1217375.}
\and
Graham Cormode%
\thanks{AT\&T Labs---Research.}
\and
Navin Goyal%
\thanks{Microsoft Research India.}
\and
Justin Thaler%
\thanks{School of Engineering and Applied Sciences, Harvard University. Supported by a NSF Graduate Research Fellowship and NSF grants CNS-1011840 and CCF-0915922.}
}

\newenvironment{proofof}[1][Proof: ]{\noindent \textbf{#1}}{\qed\medskip}

\newcommand{\E}{\mathbb{E}}

\newcommand{\fn}{F}

\newcommand{\cF}{\mathcal{F}}
\newcommand{\s}{m}
\newcommand{\length}{N}
\newcommand{\footprint}{M}

\newcommand{\poly}{\mathrm{poly}}

\newcommand{\ceq}{\subseteq}
\newcommand{\pq}{{\sc PointQuery}\xspace}
\newcommand{\hh}{{\sc HeavyHitters}\xspace}
\newcommand{\select}{{\sc Selection}\xspace}
\newcommand{\inject}{{\sc Injection}\xspace}
\newcommand{\subinj}{{\sc SubInjection}\xspace}
\newcommand{\multiindex}{{\sc MultiIndex}\xspace}
\newcommand{\subftwo}{{\sc Sub$F_2$}\xspace}

\DeclareMathOperator{\err}{err}
\DeclareMathOperator{\hcost}{hcost}
\DeclareMathOperator{\vcost}{vcost}
\DeclareMathOperator{\MA}{MA}

\DeclareMathAlphabet{\mathpzc}{OT1}{pzc}{m}{it}
\DeclareMathAlphabet{\mathcal}{OMS}{cmsy}{m}{n}

\newcommand{\out}[2]{\operatorname{out}^{#1}(#2)}

\newcommand{\ascheme}{scheme\xspace}
\newcommand{\aschemes}{schemes\xspace}
\newcommand{\aScheme}{Scheme\xspace}
\newcommand{\aSchemes}{Schemes\xspace}
\newcommand{\scheme}[1]{$#1$-scheme\xspace}
\newcommand{\oscheme}[1]{online $#1$-scheme\xspace}

\newcommand{\pscheme}[1]{prescient $#1$-scheme\xspace}

\newcommand{\eat}[1]{}
           
\renewcommand{\b}{\{0,1\}}
\newcommand{\F}{\mathbb{F}}

\newcommand{\help}{\mathfrak{h}}
\newcommand{\rhelp}{{\help,r_P}}

\newcommand{\alg}{\mathcal{A}}

\newcommand{\cP}{\mathcal{P}}

\newcommand{\fail}{\bot}

\newcommand{\bx}{\mathbf{x}}
\newcommand{\by}{\mathbf{y}}

\newcommand{\disj}{\textsc{disj}\xspace}
\newcommand{\kdisj}{$\s$-\textsc{disj}\xspace}
\newcommand{\idx}{\textsc{index}\xspace}
\newcommand{\fk}{\ensuremath{F_k}}
\newcommand{\universe}{{\mathcal U}}

\newcommand{\etal}{{\textit{et al.}}\xspace}

\newcommand{\ignore}[1]{}
\newcommand{\provisionallyremove}[1]{}

\newcommand{\anncost}{c_a}

\newtheorem{theorem}{Theorem}[section]

\newtheorem{lemma}[theorem]{Lemma}
\newtheorem{corollary}[theorem]{Corollary}

\newtheorem{proposition}[theorem]{Proposition}
\theoremstyle{definition}
\newtheorem{definition}[theorem]{Definition}
\newtheorem{remark}{Remark}

\setlist[enumerate]{topsep=5pt,itemsep=1pt}
\makeatletter
\newcommand\squeezepar{\@startsection{paragraph}{4}{\z@}{1.5ex \@plus1ex \@minus.2ex}{-1em}{\normalfont\normalsize\bfseries}}
\makeatother
\renewcommand{\paragraph}[1]{\squeezepar{{#1}}}

\begin{document}
 
\maketitle

\begin{abstract}

Motivated by the surging popularity of commercial cloud computing services, a
number of recent works have studied \emph{annotated data streams} and variants
thereof.
%\cite{icalp, esa, vldb, itcs, raz, klauck}. 
In this setting, a computationally weak \emph{verifier} (cloud user), lacking
the resources to store and manipulate his massive input locally, accesses a
powerful but untrusted \emph{prover} (cloud service).  The verifier must work
within the restrictive data streaming paradigm.  The prover, who can
\emph{annotate} the data stream as it is read, must not just supply the final
answer but also convince the verifier of its correctness. Ideally, both the amount
of annotation from the prover and the space used by the verifier should be
sublinear in the relevant input size parameters.

A rich theory of such algorithms---which we call \emph{schemes}---has started
to emerge. Prior work has shown how to leverage the prover's power to
efficiently solve problems that have no non-trivial standard data stream
algorithms.  However, even though optimal schemes are now known for several
basic problems, such optimality holds only for streams whose length is commensurate with the
size of the \emph{data universe}.  In contrast, many real-world data sets are
relatively \emph{sparse}, including graphs that contain only $o(n^2)$ edges, and IP
traffic streams that contain much fewer than the total number of possible IP
addresses, $2^{128}$ in IPv6.

Here we design the first annotation schemes that allow both the annotation and
the space usage to be sublinear in the total number of stream {\em updates}
%\emph{stream length} $m$, 
rather than the size of the data universe.  We solve significant problems,
including variations of \idx, \textsc{set-disjointness}, and
\textsc{frequency-moments}, plus several natural problems on graphs. On the
other hand, we give a new lower bound that, for the first time, rules out
smooth tradeoffs between annotation and space usage for a specific problem.
Our technique brings out new nuances in Merlin--Arthur communication
complexity models, and provides a separation between online versions of the MA
and AMA models. 

%We also prove a new lower bound for on the online MA communication of a natural problem, the sparse Point Query Problem.
%The lower bound demonstrates that the online MA communication complexity of the sparse Point Query Problem is essentially equal its one-way communication complexity,
%thereby ruling out smooth tradeoffs between 
%annotation size and space usage for any annotations protocol. This is the first
%problem that provably exhibits this phenomenon.  
%We consider problems ranging from simple reporting queries like INDEX and HEAVY HITTERS,
%to complicated aggregation queries like FREQUENCY MOMENTS.
%These therefore yield the first non-trivial protocols in this model for sparse data streams.
%A few of our protocols (e.g. INDEX, SELECTION, and HEAVY HITTERS) are provably optimal up to logarithmic factors.
%\mnote{Do we need to mention MA/AM/AMA here?}

\end{abstract}

\section{Introduction} \label{sec:intro}

The surging popularity of commercial cloud computing services has rendered the following scenario
increasingly plausible.
A business---call it AliceSystems---processes billions or trillions of transactions a day. 
The volume is sufficiently high that AliceSystems cannot or will not
store and process the transactions on its own. Instead, it 
offloads the processing to a commercial cloud computing service.

The offloading of any computation raises issues of trust. AliceSystems may be concerned
about relatively benign errors: perhaps the cloud dropped some of the transactions, executed a buggy algorithm,
or experienced an uncorrected hardware fault. Alternatively, AliceSystems
may be more cautious
and fear that the cloud operator is deliberately deceptive or has been externally compromised.
Either way, each time AliceSystems poses a query to the cloud, it may demand that the cloud provide not only the answer but also some proof that the returned answer is correct.

Motivated by this scenario, a number of recent works have studied annotated data streams and their
variants~\cite{icalp, esa, itcs, vldb, raz, klauck}. 
In this setting, a computationally weak \emph{verifier} (modeling AliceSystems in the above scenario),
who lacks the resources to store the entire input locally, is 
given access to a powerful but untrusted \emph{prover} (modeling the cloud computing service).
The verifier must execute within the confines of the restrictive \emph{data streaming} paradigm, 
i.e., it must process the input sequentially in whatever order it arrives, using space that is substantially sublinear 
in the total size of the input. The prover is allowed to annotate the data stream as it is read, with the goal
of convincing the verifier of the correct answer.
The streaming restriction for the verifier fits the cloud computing setting well, as the verifier's streaming pass
over the input can occur while uploading data to the cloud.

Prior work \cite{AW, icalp, esa, itcs, klaucksolo, klauck} has provided considerable
understanding of the power of annotated data streams, revealing a surprisingly rich theory.
A number of fundamental problems that possess no non-trivial algorithms
in the standard streaming model do have efficient \emph{schemes} 
when the data stream may be annotated by a prover: the term ``scheme'' 
refers to an algorithm involving verifier-prover interaction as above.
By exploiting powerful algebraic techniques originally developed
in the literature on interactive proofs~\cite{LundFKN92,Shamir92}, these works have achieved
essentially optimal tradeoffs between annotation size and the space usage of the verifier
for problems ranging from frequency moments to bipartite perfect matching. 

However, these schemes are only optimal for streams for which the 
total number of updates is large relative to the size of the \emph{data universe}. 
In contrast, many real-world data sets are \emph{sparse}: for example, many real-world graphs, though large, contain much fewer than the maximum possible number $\binom{n}{2}$ of edges, and IP traffic streams contain much fewer than the total number of possible IP addresses, $2^{128}$ in IPv6.

In this paper, we give the first schemes in the annotations model that
allow both the annotation size and space usage to be 
\emph{sublinear in %the \emph{stream length} $m$,
the number of items with non-zero frequency in the data stream}, rather
than the size of the data universe $n$.
%We consider problems ranging from simple reporting queries like INDEX and HEAVY HITTERS,
%to complicated aggregation queries like FREQUENCY MOMENTS.
On the negative side,
we also give a new lower bound that for the first time rules out smooth tradeoffs between 
annotation size and space usage for a specific problem. The latter result is
derived from a new lower bound in the Merlin--Arthur (MA) communication model
that may be of independent interest.

\subsection{Related Work} \label{sec:related}

Aaronson and Wigderson \cite{AW} gave a beautiful MA communication protocol
for the \textsc{set-disjointness} problem (henceforth, $\disj$) using
algebraic techniques analogous to those in the famous ``sum-check protocol''
from the world of interactive proofs and probabilistically checkable proofs \cite{LundFKN92}.
Their protocol is nearly optimal, essentially matching a lower bound of Klauck
\cite{klaucksolo}.  The Aaronson--Wigderson protocol has served as the
starting point for many schemes for annotated data streams. We will refer to
such schemes as \emph{sum-check schemes}\/; a typical example is
Proposition~\ref{prop:dense} in this work.

Aaronson \cite{aaronson} studied the hardness of the \idx problem in a restricted version of the MA communication model, as well as in a quantum variant of this model. 
His classical model is similar to the online MA communication model that we consider.
Annotated data streams were introduced by Chakrabarti \etal\ \cite{icalp}, and studied further by Cormode \etal\ \cite{esa}.
These two papers gave essentially optimal annotation schemes for problems ranging from exact computation of Heavy Hitters and Frequency Moments
to graph problems like Bipartite Perfect Matching and Shortest $s$-$t$
Path. Cormode, Thaler and Yi~\cite{vldb} later extended the annotations
model to allow the prover and verifier to have a \emph{conversation}, and dubbed
this interactive model \emph{streaming interactive proofs}. They demonstrated that  streaming interactive proofs
can have exponentially smaller space and communication costs than annotated data streams, and showed that
a number of powerful protocols from the literature on interactive proofs can be made to work with streaming verifiers; 
in particular, this applies to a powerful general-purpose interactive proof protocol due to Goldwasser, Kalai, and Rothblum \cite{Goldwasser:Kalai:Rothblum:08}. 
Cormode, Mitzenmacher, and Thaler \cite{itcs} implemented a number of protocols in both the annotated data streams and
streaming interactive proof settings, demonstrating genuine scalability in many cases. In particular, they developed an
implementation of the Goldwasser \etal protocol \cite{Goldwasser:Kalai:Rothblum:08}
that approaches practicality. Most relevant to our work on annotated data streams, Cormode, Mitzenmacher, and Thaler also used sophisticated FFT 
algorithms to drastically reduce the prover's runtime in the sum-check schemes,
which we make frequent use of. 

Two recent works have considered variants of the annotated data stream model. Klauck and Prakash \cite{klauck} study a restricted version of the annotations model in which 
the annotation must essentially end by the final stream update. 
Gur and Raz \cite{raz} give protocols for a class of problems in a model that is similar to annotated data streams, but more powerful in that the verifier
has access to both public and private randomness. 
This corresponds to the AMA communication model. 
We consider protocols in this model in Section \ref{sec:amastream}.

Early work on interactive proof systems studied the power of space-bounded verifiers (the survey by Condon \cite{condon} 
provides a comprehensive overview), but many of the protocols
developed in this line of work require the verifier to store the input, and therefore do not work 
in the annotations model, where the verifier must be streaming. An exception is work by Lipton~\cite{lipton}, who
relied on using fingerprinting techniques to allow a log-space streaming verifier to ensure that the prover correctly plays back the transcript 
of an algorithm in an appropriate computational model. 
This approach does not lead to
protocols with sublinear annotation length.
More recently, Das Sarma \etal\ studied the ``best order streaming model,'' which can be thought of 
as the annotations model where the annotation is restricted to be a permutation of the input \cite{bestorder}.

%Another direction followed in recent years has been to apply the power
%of interactive proofs to tractable problems. 
%Sparked by the work of Goldwasser {\em et
 % al.}~\cite{Goldwasser:Kalai:Rothblum:08}, there has been a line of
%work on verifying circuit based computations. 
%Although the work of the verifier can typically be performed in a
%single streaming pass over the input and the proof, we restrict our
%attention to protocols which have a constant number of rounds. 
%In contrast, these protocols have many rounds, typically
%logarithmically many rounds for each level of the circuit considered. 

\subsection{Overview of Results and Techniques}

We give an informal overview of our results and the techniques we use to obtain them.
Throughout, $n$ will denote the size of the data universe and $\s$ the number of items with non-zero frequency at the end
of a data stream (we refer to $\s$ as the ``sparsity'' of the stream). A scheme in which the streaming 
verifier uses at most $c_v$ bits of storage and requires at most $\anncost$ bits of annotation from the prover is called a $(\anncost, c_v)$-scheme. Section~\ref{sec:models} defines our models of computation carefully and sets up terminology.
\medskip

{\bf Section \ref{sec:pointquery}} contains our first set of results. We begin by precisely characterizing the complexity of the sparse \pq\ problem---a natural variant of the well-known \idx\ problem from 
communication complexity---giving an $(x \log n, y\log n)$-scheme whenever
$xy \ge \s$. We give similar upper bounds for the related problems \select\ and \hh.
We also prove a lower bound showing that \emph{any} $(c_a, c_v)$-scheme for these problems requires $c_a c_v = \Omega(\s \log(n/\s))$, improving by a $\log(n/\s)$
factor over lower bounds that follow from prior work on ``dense'' streams. By a dense stream we mean one where $n$ is not much larger than $\s$. This $\log(n/\s)$ factor may seem minor, but a
striking consequence is that the (very) sparse \idx\ problem---where Alice's $n$-bit string has
Hamming weight $O(\log n)$---has one-way randomized communication complexity that is 
within a logarithmic factor of its online MA communication complexity.
This implies that  no non-trivial tradeoffs between Merlin's and Alice's message sizes
are possible for this problem; to our knowledge this is the first
problem that provably exhibits this phenomenon. 

%The key technique in o
Our scheme for sparse \pq\ 
relies on universe reduction: the prover succinctly describes a mapping $h: [n] \rightarrow [r]$ that
maps the input stream, which is defined over the huge data universe $[n]$,
down to a derived stream defined over a smaller universe $[r]$. By design,
if the prover is honest and the mapping $h$ does not cause ``too many collisions,'' then the answer
on the original
stream can be determined from the answer on the derived stream. We then efficiently apply 
known \aschemes\ for dense streams to the derived stream.
 
For our lower bound in Section \ref{sec:pointquery}, 
we give a novel reduction from the standard (dense) \idx\ problem to sparse \idx\ 
that is tailored to the MA communication model. We then apply known lower bounds for dense \idx.
Our technique also gives what is to our
knowledge the first polynomial separation between the online MA and
AMA communication complexities of a specific (and natural) problem.
\medskip

For clarity, the remainder of this overview omits factors logarithmic in $n$ and $\s$ when stating the costs of schemes. Though these factors are
important for Section~\ref{sec:pointquery} (the consequences of our lower bound being most significant when $n = \s^{\omega(1)}$), we anticipate
that in practice $n$ and $\s$ will usually be polynomially related.  \medskip

{\bf Sections \ref{sec:disj} and \ref{sec:of2}} contain our most interesting and technically involved results, namely, efficient schemes for \textsc{size}-$\s$-\textsc{set-disjointness} (henceforth, \kdisj) and $k$th Frequency Moments (henceforth, \fk). 
The schemes here are substantially more complex than those in
Section \ref{sec:pointquery} and represent the main technical contributions of this paper.

Section \ref{sec:disj} gives $(\s^{2/3}, \s^{2/3})$-schemes for both
problems, but the schemes rely on 
``prescient'' annotation, i.e., annotation
provided at the start of the stream that depends on the stream itself. 
The even more complex schemes of Section \ref{sec:of2} eliminate the need for prescient annotation 
and also achieve much more general tradeoffs between annotation length and space usage. 
Specifically, Section \ref{sec:of2} gives $(\s c_v^{-1/2}, c_v)$-schemes for $m$-\disj and \fk\ for any 
$c_v < \s$. Notice that one recovers the costs achieved in Section \ref{sec:disj} by setting
$c_v=\s^{2/3}$.

These schemes are the first for these problems that allow both the annotation length and space usage to be sublinear in $\s$.
At a very high level, there are three interlocking ideas that allow us to 
achieve this.

\begin{enumerate}
\item The first idea is a careful application of universe
  reduction. 
We were able to use a simple version of this idea to derive the upper bound for the \pq\ problem in Section \ref{sec:pointquery},
but in the case of \disj\ and \fk\, the universe-reduction mapping $h:
[n] \rightarrow [r]$ specified by the prover is more complicated, 
and requires refinement in the form of the additional ideas described
below. 

\item The second idea is addressed to ensuring that the prover performed the universe-reduction step
in an honest manner, in the sense that the answer on the original stream can indeed be determined from the answer
on the derived stream. The difficulty of ensuring $P$ is honest varies depending on the structure of the problem at hand.
For $\fk$, the verifier has to make sure that the universe-reduction
mapping $h$ is injective on the items appearing in the data stream. This requires developing an efficient way for $V$
to detect collisions under $h$, even though $V$ does not have the space to store all of the values $h(x_i)$ for stream updates $x_i$.
For \kdisj, a notion weaker than injectiveness is sufficient. %, which allows us to achieve costs sublinear in
%the stream sparsity $M$ rather than the stream length $m$.

\item The third idea pertains to allowing $P$ to specify the universe-reduction mapping $h$ \emph{online}.
That is, for many problems it would be much simpler if $P$ could 
determine the mapping $h$ in advance i.e. if $P$ could be prescient,
and send $h$ to $V$ at the start of the stream so that $V$ can determine the derived ``mapped-down'' stream on her own (this is the approach taken in Section \ref{sec:disj}).
When $P$ must specify $h$ in an online fashion, additional insight is required. At a high level, our approach is to have $P$
specify a ``guess'' as to the right hash function at the beginning of the steam, and retroactively modify the hash function after 
the stream has been observed. The challenging aspect of this approach is to ensure that $P$'s retroactive modification of the hash function
is consistent with the observed data stream, even though $V$ cannot refer back to the stream to enforce this.%(rather than using perfect hash families as a black box, one must recognize that a perfect hash family can be constructed via a pairwise-independent hash function followed by a list of 'exceptions'). 
%In addition, one has to make sure that $P$ 
%is \emph{consistent} i.e. that $P$ does not try to get $V$ to use different universe-reduction mappings on different parts of the stream.
%How difficult it is to achieve these goals depends on the problem.

We exploit similar ideas to allow $V$ to avoid storing the universe-reduction mapping $h$ herself; this is the key to achieving general tradeoffs between 
annotation length and space usage in Section \ref{sec:of2}.
In some \aschemes, storing this mapping $h$ would be the bottleneck in $V$'s space usage. We show how 
$V$ can store only a \emph{partial} description of $h$, and ask $P$ to fill in the remainder of the description
when necessary. %As above, the challenging aspect is in ensuring that $P$ is consistent; i.e. that $P$ does not try to get $V$
%to use a hash function that is inconsistent with the observed data stream.

\end{enumerate}
\medskip

{\bf Section \ref{sec:graphs}} exploits all of these results, applying them to
several graph problems, including
counting triangles and demonstrating a perfect matching. Our schemes have costs
that depend on the number of edges in the graph, rather than the total number of possible edges,
and demonstrate that the ideas underlying our
\kdisj\ and \fk\ schemes are broadly applicable. We state clearly how our schemes improve over prior work throughout.
\medskip

{\bf Section \ref{sec:genturnstile}} considers a more general stream update 
model, which allows items to have negative frequencies.
These negative frequencies potentially break the ``collision detection'' sub-protocol used in the
previous sections, so we show how to exploit a source of public randomness to
allow these protocols to be carried out. 
Essentially, the public randomness specifies a remapping of the input,
so that the prover is highly unlikely to be able to use negative frequencies to ``hide''
collisions. Because the protocols of Section \ref{sec:genturnstile} require public randomness,
they work in the AMA communication and streaming models, as opposed to the MA models in which all of our other protocols operate.
%In contrast to previous results, here we do need to use
%truly random values, corresponding to an AMA model. 
%Previous results in the MA model could tolerate $P$ providing
%supposedly random choices, as failure to do so would only hurt $P$'s
%claims. 

\section{Models, Notation, and Terminology} \label{sec:models}

Many of the algorithms (schemes) in this paper use randomization in subtle
ways, making it important to properly formalize several models of computation.
We begin with Merlin--Arthur communication models, a topic first studied by
Babai, Frankl and Simon~\cite{BabaiFS86}, which we eventually use to derive
lower bounds. We then turn to annotated data stream models. At the end of the
section we set up some notation and terminology for the rest of the paper.
Some of our discussion in this section borrows from prior work \cite{icalp}. 

\subsection{Communication Models}\label{sec:communicationmodel}

Let $\fn : X \times Y \rightarrow \b$ be a function, where $X$ and $Y$ are
both finite sets.  This naturally gives a $2$-player number-in-hand
communication problem, where the first player, Alice, holds an input $x \in
X$, and the second player, Bob, holds an input $y \in Y$. The players wish to
compute $\fn(x, y)$ by executing a (possibly randomized) communication
protocol that correctly outputs $\fn(x,y)$ with ``high'' probability. In
Merlin--Arthur communication, there is additionally a ``super-player,'' called
Merlin, who knows the entire input $(x, y)$, and can help Alice and Bob by
interacting with them. The precise pattern of interaction matters greatly and
gives rise to distinct models.  Merlin's goal is to get Alice and Bob to
output ``$1$'' regardless of the actual value of $\fn(x,y)$, and so Merlin is
not to be blindly trusted.

One important departure we make from prior work is that {\em we allow Merlin
to use private random coins} during the protocol. Most prior work on MA (and AM)
communication~\cite{BabaiFS86,klaucksolo,klaucksolo2} defined Merlin to be
deterministic, which does not make a difference in the basic setting.  But in
this work we are concerned with ``online MA'' models, where the distinction
does matter, and these online MA models are in close correspondence with the
annotated data stream models that are our eventual topic of study.

%, with high probability. 
%For simplicity, we
%assume that the players communicate using a blackboard. 
%\comment{am: can we instead assume private messages. this makes defining the data stream models much simpler.}

\paragraph{MA Communication.} 
In a Merlin--Arthur protocol (henceforth, ``MA protocol'') for $\fn$, Merlin
begins by sends a help message $\help(x,y,r_M)$, using a private random string
$r_M$, that is seen by both Alice and Bob. Then Alice and Bob (the pair that
constitutes the entity ``Arthur'') run a randomized communication protocol
$\cP$, using a public random string $r_A$, eventually outputting a bit
$\out{\cP}{x, y, r_A, \help}$.  Importantly, $r_A$ is not known to Merlin at
the time he sends $\help$.  The protocol $\cP$ is $\delta_s$-sound and
$\delta_c$-complete if there exists a function $\help:X \times Y \times \b^*
\to \b^*$ such that the following conditions hold.
\begin{enumerate}
  \item If $\fn(x, y) = 1$ then 
  $\Pr_{r_M, r_A}[\out{\cP}{x, y, r_A, \help(x, y, r_M)} = 0] \le \delta_c$.
  \item If $\fn(x, y) = 0$ then $\forall\, \help'\in\b^*:\, 
  \Pr_{r_A}[\out{\cP}{x, y, r_A, \help'} = 1] \le \delta_s$.
\end{enumerate}

We define $\err(\cP)$ to be the minimum value of $\max\{\delta_s, \delta_c\}$
such that the above conditions hold. Following~\cite{icalp}, we define the
{\em help cost} $\hcost(\cP)$ to be $1 + \max_{x,y,r_M} |\help(x, y, r_M)|$
(forcing $\hcost \ge 1$, even for traditional Merlin-free protocols), and the
{\em verification cost} $\vcost(\cP)$ to be the maximum number of bits
communicated by Alice and Bob over all $x, y$ and $r_A$.  We define
$\MA_\delta(\fn) = \min\{\vcost(\cP) + \hcost(\cP) : \cP$ is an MA protocol
for $\fn$ with $\err(\cP) \le \delta\}$, and $\MA(\fn) = \MA_{1/3}(\fn)$.

\paragraph{Online MA Communication.}
An online MA protocol is defined to be an MA protocol, as above, but with the
communication pattern required to obey the following sequence. (1) Input $x$
is revealed to Alice and Merlin; (2) Merlin sends Alice a help message
$\help_1(x, r_M)$ using a private random string $r_M$; (3) Input $y$ is
revealed to Bob; (4) Merlin sends Bob a help message $\help_2(x, y, r_M)$; (5)
Alice sends a public-coin randomized message to Bob, who then gives a $1$-bit
output.  We see this model as the natural MA variant of one-way communication,
and the analogy with the gradual revelation of a streamed input should be
obvious.

For such a protocol $\cP$, we define $\hcost(\cP)$ to be $1 + \max_{x,y,r_M}
(|\help_1(x, r_M)| + |\help_2(x,y,r_M)|)$  We define soundness, completeness,
$\err(\cP)$, and $\vcost(\cP)$ as for $\MA$.  Define $\MA^\to_\delta(\fn) =
\min\{\hcost(\cP) + \vcost(\cP): \cP$ is an online MA protocol for $\fn$ with
$\err(\cP) \le \delta\}$ and write $\MA^\to(\fn) = \MA^\to_{1/3}(\fn)$.

\paragraph{Online AMA Communication.} 
An online AMA protocol is a souped-up version of an online MA protocol, where
public random coins can be tossed at the start, before any input is revealed.
The number of such coin tosses is added to the vcost of the protocol. This
models the cost of an initial round of communication between Arthur (i.e.,
Alice + Bob) and Merlin.  Note that the {\em second} public random string,
used when Alice talks to Bob, does not count towards the vcost.

\paragraph{On Merlin's Use of Randomness.}
In an MA protocol, Merlin can deterministically choose a help message that
maximizes Arthur's acceptance probability. However, Merlin cannot do so in the
online MA model, because he does not know the entire input when he talks to
Alice. This is why we allow Merlin to use randomness in these definitions. 

Two recent papers~\cite{icalp,klauck} use ``online MA'' to mean a more
restrictive model where a deterministic Merlin talks only to Bob and not to
Alice. With Merlin required to be deterministic, this communication
restriction is irrelevant, as Merlin cannot tell Alice anything she does not
already know. However, we permit Merlin to be probabilistic, and in this case
we do not know that Merlin can avoid talking to Alice.

As noted earlier, our goal in defining the communication models this way is to
closely correspond to annotated data stream models. In many of our online
schemes (see, e.g., Section~\ref{sec:of2}), the helper provides initial
annotation that specifies a random ``hash'' function, $h$, and the
completeness guarantee of the subsequent protocol depends crucially on $h$
having ``low collision'' properties. Since $h$ must be chosen without seeing
all of the input, such low collision properties cannot be guaranteed by
picking a fixed $h$ in advance.  However, if the helper chooses $h$ at random,
then we do have such guarantees for each fixed input, with high probability.

\subsection{Data Stream Models}\label{sec:streammodel}

We now define our annotated data stream models.  Recall that a (traditional)
data stream algorithm computes a function $\fn$ of an input sequence
$\bx \in\universe^\length$, where $\length$ is the number of stream updates, and $\universe$ is some data universe, such as 
$\b^b$
or
$[n] = \{0, \dots, n-1\}$: the algorithm uses a limited amount of working memory and has
access to a random string. The function $\fn$ may or may not be Boolean.
%When it is not, we may permit an approximate answer: we accept any answer in the set $C(f(\bx))$, for some function $C$. 

An annotated data stream algorithm, or a {\em \ascheme}, is a pair $\alg =
(\help,V)$, consisting of a help function $\help:\universe^\length \times \b^*
\to \b^*$ used by a {\em prover} (henceforth, $P$) and a data stream algorithm
run by a {\em verifier}, $V$. Prover $P$ provides $\help(\bx,r_P)$ as
annotation to be read by $V$.  We think of $\help$ as being decomposed into
$(\help_1,\ldots,\help_{\length})$, where the function
$\help_i:\universe^{\length}\to\b^*$ specifies the annotation supplied to $V$
after the arrival of the $i$th token $x_i$. That is, $\help$ acts on
$\bx$ (using $r_P$) to create an {\em annotated stream} $\bx^\rhelp$ defined as follows:
\[
  \bx^{\rhelp} :=
  (x_1,\, \help_1(\bx,r_P),\,
    x_2,\, \help_2(\bx,r_P),\, \ldots,\,
    x_{\length},\, \help_{\length}(\bx,r_P)) \, .
\]
Note that this is a stream over $\universe \cup \b$, of length $\length +
\sum_i |\help_i(\bx,r_P)|$.  The streaming verifier $V$, who uses $w$ bits of
working memory and has oracle access to a (private) random string $r_V$, then
processes this annotated stream, eventually giving an output
$\out{V}{\bx^{\rhelp},r_V}$.

% For ease of exposition, as well as consistency with the literature on interactive proofs,
% we will regularly refer to the help function $\help$ as a \emph{prover}, which we denote by $P$.
% Likewise, we will refer to the data stream algorithm $V$ as the \emph{verifier}.

\paragraph{Prescient \aSchemes.}

The \ascheme\ $\alg = (\help,V)$ is said to be $\delta_s$-sound and
$\delta_c$-complete for the function $\fn$ if the following conditions hold:
\begin{enumerate}
  \item For all $\bx\in\universe^\length$, we have $\Pr_{r_P,
  r_V}[\out{V}{\bx^{\rhelp},r_V} \neq \fn(\bx)] \le \delta_c$.
  \item For all $\bx\in\universe^\length$, $\help'=(\help_1', \help_2',
  \ldots, \help_\length') \in (\b^*)^\length $, we have
    $\Pr_{r_V}[\out{V}{\bx^{\help'},r_V} \not\in \{\fn(\bx)\} \cup\{\fail\}] \le \delta_s$.
\end{enumerate}
If $\delta_c=0$, the \ascheme\ satisfies \emph{perfect completeness};
otherwise it has \emph{imperfect completeness}. An output of ``$\fail$''
indicates that $V$ rejects $P$'s claims in trying to convince $V$ to output a
particular value for $\fn(\bx)$.  

We note two important things. First, the definition of a scheme allows the
annotation $\help_{i}(\bx,r_P)$ to depend on the entire stream $\bx$, thus
modeling {\em prescience}: the advice from the prover can depend on data which
the verifier has not seen yet.  Second, $P$ must convince $V$ of the value of
$\fn(\bx)$ {\em for all} $\bx$.  This is stricter than the traditional
definitions of interactive proofs and MA communication complexity (including
our own, above) for decision problems, which place different requirements on
the cases $\fn(\bx) = 0$ and $\fn(\bx) = 1$.  In Section \ref{sec:graphs}, we
briefly consider a relaxed definition of schemes that is in the spirit of the
traditional definition.

We define $\err(\alg)$ to be the minimum value of $\max\{\delta_s, \delta_c\}$
such that the above conditions are satisfied. We define the {\em annotation
length} $\hcost(\alg) = \max_{\bx, r_P} \sum_i |\help_{i}(\bx,r_P)|$, the
total size of $P$'s communications, and the {\em verification space cost}
$\vcost(\alg) = w$, the space used by the verifier $V$.  We say that $\alg$ is
a \pscheme{(\anncost, c_v)} if $\hcost(\alg) = O(\anncost)$, $\vcost(\alg) =
O(c_v)$ and $\err(\alg) \le \frac13$.

\paragraph{Online \aSchemes.}
We call $\alg = (\help,V)$ a $\delta$-error online scheme for $\fn$ if, in
addition to the conditions in the previous definition, each function $\help_i$
depends only on $(x_1, \ldots, x_i)$.  We define error, hcost, and vcost as
above and say that $\alg$ is an {\em \oscheme{(\anncost, c_v)}} if
$\hcost(\alg) = O(\anncost)$, $\vcost(\alg) = O(c_v)$, and $\err(\alg) \le
\frac13$.

Unlike prior work \cite{icalp}, we do not always assume that the universe size
$n$ and stream length $\length$ are polynomially related; it is possible that
$\log \length = o(\log n)$.  Therefore we must be much more careful about
logarithmic factors than in prior work.  We do assume that $\length < n$
always, because our focus is on sparse streams.

Notice that the help function can be made deterministic in a prescient scheme,
but not necessarily so in an online scheme. This is directly analogous to the
situation for MA and online MA communication models, as discussed at the end
of Section~\ref{sec:communicationmodel}.

\paragraph{AMA Schemes.} 
We also consider what we call AMA schemes, where there is a common source of
public randomness, in addition to the verifier's private random coins.  The
AMA \ascheme\ model is identical to the one considered by Gur and Raz
\cite{raz}, who referred to  it as the ``Arthur--Merlin streaming model.''

An online AMA scheme is identical to a (standard) online scheme,
except that the data stream algorithm and help function both have access to a source of public random bits.
The number of random bits used is also counted in both the $\hcost$ and the $\vcost$ of the scheme.

\paragraph{On Practicality and the Plausibility of Prescience.}
Although our definition of a scheme allows annotation to be sent after each
stream update, all the schemes we in fact design in this paper only require
annotation before the start or after the end of the stream.  As a practical
matter, this avoids the need for fine-grained coordination between the
annotation and the data stream.

Online annotation \aschemes\ have the appealing property that the prover need
not ``see into the future'' to execute them; at any time $t$, the prover's
message only depends on stream updates that arrived before time $t$.  While
the online restriction appears most natural, prescient schemes may still be
suitable in some settings, such as when $P$ has already seen the full input
prior to $V$ beginning to read it. 
%because the interpretation of prescience as requiring $P$ to see the future
%does not  necessarily capture all applications. 
Consider a volunteer computing scenario where the verifier farms out many
computations to volunteers, and only inspects a particular input if a
volunteer has already looked at that input and claims to have found something
interesting\footnote{See, for example, \url{http://boinc.berkeley.edu/}.}.  In
brief, in some settings the prover may naturally see the input before the
verifier, and in this case a prescient scheme will be feasible.

\subsection{Relationship Between MA Protocols and \aSchemes}

Any prescient (resp. online) $(\anncost,c_v)$-scheme $\alg = (\help, V)$ for a
function $\fn$ can be converted into an MA (resp.~online MA) protocol for
$\fn$ in the natural way: Merlin sends the output of the $i$th help function
$\help_{i}$ to Alice---who receives a prefix of the input stream---or Bob,
depending on which of the players possesses the $i$th piece of the input.
Alice runs the streaming algorithm $V$ on her input as well as any annotation
she received, and sends the state of the algorithm to Bob.  Bob uses this
state to continue running $V$ on his input and the annotation he received, and
then outputs the end result.  The $\hcost$ of this protocol is at most
$\anncost \log \length$, since Merlin has to specify which stream update $i$
each piece of annotation is associated with, and the $\vcost$ of this protocol
is at most $c_v$. Thus, lower bounds on usual (resp. online) MA communication
protocols imply related lower bounds on the costs of prescient (resp. online)
annotated data stream algorithms.

\subsection{Additional Notation and Terminology} \label{sec:notation}

A data stream specifies an input $\bx$ incrementally. 
Typically, $\bx$ can be thought of as a vector (although more generally it may represent a graph or a matrix). 
Each update in the stream is of the form $(i, \delta)$
where $i \in \universe$ identifies an element of the universe, 
and $\delta \in \mathbb{Z}$ describes the change to the frequency of $i$. 
The frequency of universe item $i$ is defined as $f_i(\bx) := \sum_{(j_k, \delta_k) \in \bx: j_k = i} \delta_k$. 
We refer to the vector $f(\bx)=(f_1(\bx), \dots, f_n(\bx))$ as the
\emph{frequency vector} of $\bx$, where $n$ denotes the size of the data universe.

We consider several different update models. In the most general update model, the \emph{non-strict turnstile model}, 
the $\delta$ values may be negative, and so 
$f_i$ may also be negative. In the \emph{strict turnstile model}, the $\delta$ values may be negative, but it is assumed 
that the frequencies $f_i$ always remain non-negative. 
In the \emph{insert-only model}, the $\delta$ values must be non-negative.
Orthogonal to these, 
in the \emph{unit-update} version of each model, the $\delta$ values
are assumed to have absolute value 1.
Each of our results applies to a subset of these models, and
we specify within the statement of each theorem which update models it applies to.
%Our results apply to this case, unless stated otherwise. 

Throughout, $n$ will denote the size of the data universe, $\length$ will denote the total number of stream updates, $\s$ will denote the total number of items with non-zero frequency at the end of the stream,
and $\footprint$ will refer to the total number of distinct items that ever appear within some stream update.
We will refer to $\length$ as the \emph{length} of the stream, to $\s$ as \emph{sparsity} of the stream, and to $\footprint$ as the \emph{footprint} of the stream. 
Notice that it is always the case that $\s \leq \footprint \leq \length$. In the case of insert-only streams, $\s=\footprint$, but for streams in the (strict or general) turnstile models it is possible for $\s$ to be much smaller than $\footprint$. 
 %This is an important point: the costs of some of of our protocols is in terms of the number of (distinct) items that ever appear in any stream update since 
 %any such item can contribute to a ``collision'' when we execute a universe-reduction step. 
Note also that while we talk about ``sparse'' streams, this refers to the relative size of $n$ and $\s$, not the
absolute size. Indeed, we assume that $\s$ is typically large, too
large for $V$ to store the stream explicitly (else the problems can become trivial). 

We often make use of {\em fingerprint} functions of streams, which enable a
streaming verifier to test whether two large streams have the same frequency
vector.  The verifier chooses a fingerprint function $g(\mathbf{x})$ at random
from some family of functions satisfying the property that (over the random
selection of the function $g$), 
\[ 
  \Pr[ g(\bx) = g(\by) \mid f(\bx) \neq f(\by)] < 1/p
\]
for a parameter $p$. Typically, $g(\mathbf{x})$ is an element of a finite
field of size $\poly(p)$, and hence the number of bits required to store the
value $g(\mathbf{x})$ (as well as $g$ itself) is $O(\log p)$.  Further, there
are known constructions of fingerprint functions where $g(\mathbf{x})$ can be
computed in space $O(\log p)$ by a streaming algorithm in the non-strict
turnstile update model \cite{icalp}.

\section{Point Queries, Index, Selection, and Heavy Hitters} \label{sec:pointquery}

\subsection{Upper Bounds}
Our first result is an efficient online annotation scheme for the \pq\ problem, a generalization of the familiar \idx\ problem. 

\begin{definition} In the \pq\ problem,
the data stream $\bx$ consists of a sequence of updates of the form $(i, \delta)$, followed by an 
index $\iota$, and the goal is to determine the frequency $f_{\iota}(\bx) = \sum_{(j_k, \delta_k) \in \bx: j_k = \iota} \delta_k$.
\end{definition}

A prescient $(\log n, \log n)$-\ascheme\ for this problem is trivial as $P$ can just tell $V$ the index $\iota$ at the start of the stream, and $V$
can track the frequency of $\iota$ while observing the stream.
The $\vcost$ can be improved to $O(\log \s)$ if $V$ retains a hashed value
of $\iota$, and tracks the frequency of matching updates. 
The first \ascheme\ has perfect completeness, while the second has
completeness error polynomially small in $\s$. 

The costs of the \ascheme\ below are in terms of the stream sparsity $\s$, and not the stream length $\length$ or the stream footprint $\footprint$;
this is significant if $\s \ll \footprint$, which is the case, e.g., for the well-known straggler and set-reconciliation problems that have been studied in traditional streaming and communication models \cite{straggler, setreconciliation}.
Our lower bound in Theorem \ref{thm:lb} shows our \ascheme\ is
essentially optimal for moderate universe sizes, i.e. when the universe size $n$ is sub-exponential in the sparsity $\s$.

\begin{theorem} 
\label{thm:pointquery}
For any pair $(\anncost, c_v)$ such that $\anncost \cdot c_v \geq \s$,  there is an 
\oscheme{(\anncost  \log n,c_v \log n)}
in the non-strict turnstile update model for the \pq\ problem with imperfect completeness. 
Any online $(\anncost, c_v)$ scheme with $\anncost \geq \log n$ for this problem requires $\anncost \cdot c_v = \Omega(\s \log(n/\s))$.
\end{theorem}

\begin{proof} 
$V$ requires $P$ to specify at the start of the stream a hash function $h: [n] \rightarrow [c_v]$. $V$ requires $h$ to have description length $O(\anncost )$, rejecting if this is not the case. 
We define the derived streams $\bx^j \in \universe^\length$ based on $h$: we set 
$\bx^j_k = x_k$ iff $h(x_k)=j$, and 0 otherwise. 
Intuitively, the hash function $h$ partitions the stream updates in
$\bx$ into $c_v$ disjoint buckets, and the vector $\bx^j$ describes the
contents of the $j$th bucket. 
$V$ maintains fingerprints over a field of size $\poly(n)$ of each of the 
$c_v$ different $\bx^{j}$ vectors. 

At the end of the stream, given the desired index $\iota$, $P$
provides a description %allowing $V$ to determine
of
the (claimed) frequency vector in the $h(\iota)$th derived stream,
$f(\bx^{h(\iota)})$. $V$ computes a fingerprint of the claimed frequency vector, and compares it to the fingerprint she computed from the data stream, accepting if
and only if the fingerprints match.
Since each $\bx^j$ is sparse in expectation, the cost of this description can be low: provided $h$ does not map more than $O(\anncost )$ items with non-zero frequency to $h(\iota)$, $P$
can just specify the item id and frequency of the items with non-zero frequency in $f(\bx^{h(\iota)})$. In this case, the annotation size is just 
$O(\anncost  \log n)$. 
If $P$ exceeds this amount of annotation, $V$ will halt and reject
(output $\fail$). 

%$V$ partitions the mapped-down universe $[\anncost]$ into $c_v$ regions of size $\anncost$, and for each region $V$ keeps a fingerprint over a field of size $\poly(n)$ of the \emph{un-mapped-down} items that map into the region. At the end of the stream when the index $i$ is revealed, $P$ only needs to play back the (un-mapped-down) items that mapped to the same region as $i$. As long as the hash function didn't map more than $O(h)$ items to this region, the annotation length is just $h \log n$. If $P$ tries to play back more than $O(h)$ items, $V$ will just halt and reject. 

Soundness follows from the fingerprinting guarantee: if $P$ does not honestly provide $\bx^{h(\iota)}$, 
$V$'s fingerprint of $\bx^{h(\iota)}$ computed from the data stream will not match her fingerprint of the claimed vector of frequencies. 

To show (imperfect) completeness, we study the probability that the output of an honest prover is rejected. 
This happens only if $\s(\bx^{h(\iota)})$, the number of non-zero entries in $\bx^{h(\iota)}$, is much larger than its expectation. 
By the pairwise independence of $h$, 
$\E[\s(\bx^{h(\iota)})] = \s(\bx)/c_v=\anncost $. 
Thus, by Markov's inequality, $\Pr[ \s(\bx^{h(\iota)}) > 10\anncost ] < 1/10$. 
%that $i^*$'s region receives more than $10h$ such items is at most $1/10$. Thus,
So by specifying a hash function chosen at random from a pairwise independent hash family, and then honestly playing back the items that map to the same region as $\iota$, $P$ can convince $V$ to accept with probability $9/10$. 

Notice that $V$ does not need to \emph{enforce} that $P$ picks the hash function $h$ at random from a pairwise-wise independent hash family, as $P$ has no incentive not to pick the hash functions in this way. 
That is, since $V$ will reject if too many items map to the same region as $\iota$, it is \emph{sufficient} for $P$ to pick $h$ at random from a pairwise independent hash family in order to convince $V$ to accept with constant probability. 
But it is  equally acceptable if $P$ wants to pick $h$ another way; if he does so, $P$ just risks that $V$ will reject with a higher probability.
%The lower bound follows immediately from the $\anncost c_v \geq n$ lower bound for INDEX in \cite{icalp}, since for the INDEX problem,
%the stream length $m$ is never larger than $n$ (and hence in the hard instance of \cite{icalp}, the stream length $m$
%is at most the universe size $n$).

The lower bound follows from Theorem \ref{thm:lb}, which we prove in Section \ref{sec:lb}.
\end{proof}

The \ascheme\ of Theorem \ref{thm:pointquery} yields nearly optimal \aschemes\ for the \hh\ and \select\ problems, described below.
Table \ref{tabpointquery} summarizes these results and compares to prior work.

\begin{table}[tb]
\small
\centering
\begin{tabular}{c|r@{}l|c|c|c}
%\hline
Problem & \multicolumn{2}{|c|}{\aScheme\ Costs} &  Completeness & Prescience & Source \\ \hline
\pq & $(\log n, \log n)$&& Perfect & Prescient &  \cite{icalp}\\
%\hline
%\hline
\pq & $(\s\log n, \log n)$&& Perfect & Online &  \cite{icalp}\\
\pq & $(\anncost \log n, c_v\log n)$&: $\anncost c_v \geq n$ & Perfect & Online &
\cite{icalp}\\
\pq & $(\anncost \log n, c_v\log n)$&: $\anncost c_v \geq \s$& Imperfect & Online & Theorem \ref{thm:pointquery}\\
\hline
\select & $(\anncost \log n, c_v\log n)$&: $\anncost c_v \geq n$ & Perfect & Online & \cite{icalp}\\
%\hline
\select & $(\s\log n, \log n)$&& Perfect & Online & \cite{icalp}\\
%\hline
\select & $( \anncost \log^2 n, c_v \log n)$&: $\anncost c_v \geq \s \log n$ & Imperfect & Online & Corollary \ref{cor:selection}\\
\hline
$\phi$-\hh & $(\phi^{-1} \log n, \phi^{-1} \log n)$&: $\anncost c_v \geq n$ & Perfect & Prescient & \cite{icalp}\\
$\phi$-\hh & $(\phi^{-1} \anncost \log n, c_v \log n)$&: $\anncost c_v \geq n$ & Perfect & Online & \cite{icalp}\\
%\hline
%\hline
$\phi$-\hh & $(\s \log n, \log n)$ && Perfect & Online & \cite{icalp}\\
%\hline
$\phi$-\hh & $(\phi^{-1} \anncost \log n, c_v \log n)$&: $\anncost c_v \geq \s \log n$ & Imperfect & Online & Corollary \ref{cor:hh}\\
%\hline
$\phi$-\hh & $(\phi^{-1} \log n + \anncost  \log n, c_v \log n)$&: $\anncost c_v \geq \s \log n$ & Imperfect & Online & Corollary \ref{cor:hh2}\\
\end{tabular}
\caption{Comparison of our \aschemes\ to prior work. For all three problems,
ours are the first online \aschemes\ to achieve both annotation and space usage sublinear in the stream
sparsity $\s$ when $\s \ll \sqrt{n}$, and we strictly improve over the online MA communication cost of prior \aschemes\ whenever $\s=o(n)$. For brevity,
we omit factors of $\log_{c_v}(\s)$ from the statement of costs of the $\phi$-\hh\ \ascheme\ due to Corollary \ref{cor:hh2}}
\label{tabpointquery}
%\vspace{-0.15in}
\end{table}

\subsubsection{Selection} Our definition of the \select\ problem assumes all frequencies $f_i := \sum_{(j_k, \delta_k): j_k = i} \delta_k$ are non-negative,
and so this definition is only valid for the strict turnstile update model.

\begin{definition} \label{def:select} The \select\ problem is defined
  in terms of the quantity $N = \sum_{i\in [n]} f_i$, the sum of all the frequencies. 
Given a desired rank $\rho \in [N]$, output an item $j$ from the stream 
$\bx=\langle (j_1, \delta_1), \dots , (j_m, \delta_m)\rangle$, such
that $\sum_{(j_k, \delta_k): j_k < j} \delta_k <\rho$ and $\sum_{(j_k,
  \delta_k): j_k > j} \delta_k \geq N - \rho$. 
\end{definition}

%For consistency, we assume that all frequencies $f_i := \sum_{(j_k, \delta_k): j_k = i} \delta_k$ are non-negative.

%\textbf{JT: It's much less notation heavy to define the SELECTION problem in the case of insertion-only streams as we did in \cite{icalp}, but I wanted to emphasize that we can handle deletions... the definition of the SELECTION problem I gave doesn't really make sense if frequencies can be negative though, so maybe I really should just define the selection problem

%The following result is stated in the strict turnstile model for simplicity

\begin{corollary}
\label{cor:selection}
For any pair $(\anncost , c_v)$ such that $\anncost c_v \geq \s \log n$,  there is an \oscheme{(\anncost  \log^2 n, c_v \log n)} for \select\ in the strict turnstile update model. %Any online $(h, v)$ protocol for SELECTION requires $\anncost c_v = \Omega(m)$.
\end{corollary}

The corollary follows from a standard observation to reduce 
\select\ to answering prefix sum queries, and hence to 
multiple instances of the \pq\ problem. 
$V$ treats each stream update $(i, \delta)$ in the stream $\bx$ as an update to $O(\log n)$ dyadic ranges, where a dyadic range is a range of the form $[j2^k, (j+1)2^k-1]$ for some $j$ and $k$. 
Thus, we can view the set of dyadic range updates implied by $\bx$
as a derived stream of sparsity $\s \log n$.  
Notice we are using the fact that this transformation from the original stream of sparsity $\s$ results in a derived stream of sparsity at most $\s \log n$; a different derived stream was used in \cite{icalp} to address the \select\ problem, but the sparsity of that derived stream could be substantially larger than the sparsity of the original stream.

For any $i$, the quantity $T_i := \sum_{(j, \delta): j \leq i} \delta$
can be written as the sum of the counts of $O(\log n)$ dyadic ranges. Thus, at the end of the stream $P$ can convince $V$ that item $i$ has the desired $T_i$ value by 
running $\log n$ \pq\ protocols as in Theorem \ref{thm:pointquery} in parallel on the derived stream of sparsity $\s \log n$.
The verifier's space usage is the same as for a single \pq\ instance on this stream:
 $V$ fingerprints each of the derived streams $\bx^j$ defined in the proof of Theorem \ref{thm:pointquery}, and uses these fingerprints in all $\log n$ instances
 of the \pq\ \ascheme.
 The annotation length
is $\log n$ times larger than that required for a single \pq\ instance 
because $P$ may have to describe the frequency vectors of up to $\log n$ derived streams.

Thus, we get an
\oscheme{(\anncost \log^2 n, c_v \log n)} as long as $\anncost c_v = \Omega(\s \log n)$.
%\end{proof}
%The lower bound follows from the same reduction of INDEX to SELECTION that appears in \cite[Theorem 3.2]{icalp}, as
%this reduction preserves sparsity of the stream up to constant factors.
%\end{proof}

\subsubsection{Frequent Items}
Our definition of the $\phi$-\hh\ problem also assumes all frequencies $f_i := \sum_{(j_k, \delta_k): j_k = i} \delta_k$ are non-negative,
and so this definition is only valid for the strict turnstile update model.

\begin{definition} The $\phi$-\hh\ problem (also known as frequent items) is to list those items 
$i$ such that $f_i \geq \phi N$, i.e.  
whose frequency of occurrence exceeds a $\phi$ fraction of the total count $N = \sum_{i \in [n]} f_i$.
\end{definition}

We give a preliminary result for the $\phi$-\hh\ problem in Corollary \ref{cor:hh} below. We give a substantially improved scheme in Section \ref{sec:of2} using the ideas 
underlying our online scheme for frequency moments.

\begin{corollary}
\label{cor:hh}
For all $\anncost,c_v$ such that $\anncost c_v \geq \s \log n$, there is an \oscheme{(\anncost \phi^{-1} \log n, c_v\log n)} for solving $\phi$-\hh\ in the strict turnstile update model. %Any online or prescient $(h, v)$-protocol for computing the $\phi$-heavy hitters for $\phi=1/m$ requires $\anncost c_v = \Omega(m)$.
\end{corollary}

%\begin{proof} 
Corollary \ref{cor:hh} follows from the following analysis. \cite[Theorem 6.1]{icalp} describes how to reduce $\phi$-\hh\ to demonstrating
the frequencies of $O(\phi^{-1})$ items in a derived stream. Moreover, the derived stream has sparsity $O(\s \log n)$
if the original stream has sparsity $\s$. 
We use the \pq\ \ascheme\ of Theorem \ref{thm:pointquery}.
As in Corollary \ref{cor:selection},
the annotation length blows up by a factor $\phi^{-1}$ relative to a single \pq, but the space usage of $V$ can  remain the same as in a single \pq\ instance.
Hence, we obtain an \oscheme{(\anncost\phi^{-1} \log n, c_v\log n)} for any $\anncost c_v \geq \s \log n$.
%\end{proof}

\subsection{Lower Bound}
\label{sec:lb}
In this section, we prove a new lower bound on the online MA communication complexity of the $(m, n)$-Sparse \idx\ problem.

\begin{definition} In the $(m, n)$-Sparse \idx\ problem, Alice is given a vector $x \in \{0, 1\}^n$ of Hamming weight at most $m$, and Bob is given an index $\iota$.
Their goal is to output the value $x_\iota$.
\end{definition}

We prove our lower bound by reducing the (dense) \idx\ problem (i.e. the $(m, n)$-Sparse \idx\ problem with $m=\Theta(n)$) in the MA communication model to the $(m, n)$-Sparse \idx\ problem
for small $m$.
 The idea is to replace Alice's dense input with a sparser input over
 a bigger universe, and then take advantage of our sparse \pq\ protocol.
A lower bound on the online MA communication complexity of the dense \idx\ problem was proven in \cite[Theorem 3.1]{icalp};
there, it was shown that any online MA communication protocol $\cP$ requires $\hcost(\cP) \vcost(\cP) \geq n$. 
Combining this with our reduction of the dense \idx\ problem to the sparse version,
we conclude that any protocol for sparse \idx\ must be costly.

\begin{lemma} 
\label{lemma:denseindex}\cite[Theorem 3.1]{icalp}
Any online MA communication protocol $\cP$ for the $(n, n)$-Sparse \idx\ problem must have $\hcost(\cP) \vcost(\cP) = \Omega(n)$.
\end{lemma}

\begin{remark}
The lower bound of Lemma \ref{lemma:denseindex} was originally proved by Chakrabarti \etal\ \cite{icalp} in the communication 
model in which Merlin cannot send any message to Alice. However, the proof
easily extends to our online MA communication model (where Merlin can send a message to Alice, but that message cannot depend on Bob's input). %We describe the details in Appendix \ref{app:online}.
\end{remark}

\begin{theorem}
\label{thm:denseindex}
\label{thm:lb}
Any online MA communication protocol $\cP$ for the $(m, n)$-Sparse \idx\ problem for which $\hcost(\cP) \geq \log n $ must have $\hcost(\cP) \vcost(\cP) = \Omega(m \log(n/m))$.
\end{theorem}

\begin{proof}
Assume we have an online MA communication
protocol $\cP$ for $(m, n)$-sparse \idx.
We describe how to use this online MA protocol for the sparse \idx\
problem to design one for the dense \idx\ problem on vectors of length $n'=m \log(n/m)$.

Let $k=\log(n/m)$. Given an input $x$ to the dense \idx\ problem, Alice partitions $x$ into $n'/k$ blocks of length $k$, and constructs a $0$-$1$ vector $y$
of Hamming weight $n'/k$ over the universe $\{0,1\}^{(n'/k)\cdot 2^k} = \{0,1\}^n $ %% a universe of size $(n'/k) \cdot 2^k$
as follows.
She replaces each block $B_i$  with a 1-sparse vector $v_i \in \{0, 1\}^{2^k}$, where each entry of $v_i$ corresponds to one of the $2^{k}$ possible values of block $B_i$.
That is, if block $B_i$ of $x$ equals the binary representation of the number $j \in [2^k]$, then Alice replaces block $B_i$ with the vector $e_j \in \{0, 1\}^{2^k}$, where $e_j$ denotes the vector with a 1 in coordinate $j$ and 0s elsewhere.

 Alice now has an $n'/k=m$-sparse derived input $y$ over the universe $\{0,1\}^n$.
 Merlin looks at Bob's input to see what is the index $\iota$ of the dense vector $x$ that Bob is interested in. Merlin then tells Bob the index $\ell$ such that 
 $\ell= 2^k(\iota-1)+j$, where $B_i$ is the block that $\iota$ is located in, and block $B_i$ of Alice's input $x$ equals the binary representation of the number $j \in [2^k]$. 
 Notice that Merlin can specify $\ell$ using $\log n$ bits. If Bob is convinced that $y_{\ell}=1$, then Bob can deduce the value of \emph{all} the bits in block $B_i$
 of the original dense vector $x$, and in particular, the value of $x_{\iota}$.
 
 The parties then run the assumed online MA protocol for $(m,n)$-Sparse
 \idx.  %Concretely, they can use 
%a version of the \pq\ protocol of Theorem \ref{thm:pointquery}: 
 %Alice computes the fingerprints of the blocks of $y$, and sends these to Bob, who identifies the relevant block, and compares to the fingerprint of the advice from Merlin to establish that $y_{\ell}=1$ as claimed. 
 The total $\hcost$ of this protocol is $\hcost(\cP) + \log n = O(\hcost(\cP))$, and the total $\vcost$ is $\vcost(\cP)$.
 Thus, by Lemma \ref{lemma:denseindex}, $\hcost(\cP) \vcost(\cP) = \Omega(n') = \Omega(m \log (n/m))$ as claimed.
\end{proof}

Theorem \ref{thm:lb} should be contrasted with the following well-known upper bound.

\begin{theorem} \label{thm:ub}
Assume $n < m^m$. Then the one-way randomized communication complexity of the $(m, n)$-Sparse \idx\ Problem is $O(m \log m)$.
\end{theorem}

\begin{proof}
Alice chooses a hash function $h : [n] \rightarrow [m^3]$ at random from a pairwise independent family and uses $h$ to perform ``universe reduction''. 
That is, she sends $h$ along with the set $S$ of $m$ values $\{h(j): x_j = 1\}$. Notice $h$ can be specified with $O(\log n) = O(m \log m)$ bits, and $S$ can be specified with $O(m \log m)$ bits. Bob outputs 1 if $h(\iota)  \in S$, and 0 otherwise.
The correctness of the protocol follows from the pairwise independence property of $h$: if $x_{\iota}=0$, then with high probability $\iota$ will not collide under $h$ with any $j$ such that $x_j=1$.
The total cost of this protocol is $O(m \log m)$.
\end{proof}

\subsection{Implications of the Lower Bound}

Our lower bound in Theorem \ref{thm:denseindex} has interesting consequences when it is combined with the upper bound in Theorem \ref{thm:ub}.
Consider in particular the $(m, n)$-Sparse \idx\ Problem, where $n=2^m$. Theorem \ref{thm:ub} implies that the one-way randomized communication
complexity of this problem is $O(m \log m)$; that is, without any need
of Merlin, Alice and Bob can solve the problem with $O(m \log m)$ communication.

Meanwhile, Theorem \ref{thm:lb} implies that even if Merlin's message to Bob has length $\Omega(\log n) = \Omega(m)$, 
Alice's message to Bob must have length $\Omega(m \log(n/m) / m) = \Omega(m)$. 
Indeed, Theorem \ref{thm:lb} shows that for any protocol $\cP$, if $\hcost(\cP) \geq \log n=m$, then we must have $\hcost(\cP) \vcost(\cP) =
\Omega(m \log(n/m)) = \Omega(m^2)$. In particular, this means that if $\hcost(\cP) =
m$, $\vcost(\cP)$ must be $\Omega(m)$. This trivially implies that for any protocol $\cP$
with $\hcost(\cP)$ \emph{less} than $m$, $\vcost(\cP)$ must still be $\Omega(m)$; otherwise we could
achieve a protocol with $\hcost(\cP)=m$ and $\vcost(\cP)=o(m)$ simply by running $\cP$ and adding
in extraneous bits to the proof to bring the proof length up to $m$.

Consequently, the online MA communication
complexity of this problem is at least $\Omega(m)$, which is at most a logarithmic factor smaller than the one-way randomized communication complexity.
To our knowledge, this is the first problem that provably exhibits this behavior. Specifically, this rules out smooth tradeoffs between 
annotation size and space usage in any annotated streaming protocol for the $(m, 2^m)$-Sparse \idx\ Problem.

\begin{corollary}
The one-way randomized communication complexity of the $(m, 2^m)$-Sparse \idx\ Problem is $O(m \log m)$.
The online Merlin-Arthur communication complexity is $\Omega(m)$.
\end{corollary}

%\medskip \noindent
%\textbf{Intuition.} Some intuition for the above phenomenon is as follows. Alice and Bob can run an efficient protocol without Merlin by 
%doing universe reduction on their own. Intuitively, 
 
\subsubsection{Other Sparse Problems} 

A number of lower bounds in \cite{icalp} are proved via reductions from \idx\ that preserve stream length up to logarithmic factors.
This holds for \select\ and \hh, as well as for the problem of determining the existence of a triangle in a graph. For all such problems, the lower bound of Theorem~\ref{thm:lb}
implies corresponding new lower bounds for sparse streams, i.e. streams for which $\s = o(n)$.
We omit the details for brevity.

\subsubsection{Separating Online MA and AMA Communication Complexity} 

Another implication of Theorem~\ref{thm:lb} is a polynomial separation between online MA communication complexity and 
online AMA communication complexity. 
Indeed, there is an online AMA protocol of cost $\tilde{O}(\sqrt{m})$
for the $(m, 2^{\sqrt{m}})$-Sparse \idx\ Problem, where the $\tilde{O}$ notation hides factors polylogarithmic in $m$: the first message, which consists 
of public random coins, is used to specify a hash function $h:[n] \rightarrow [m^3]$ from a pairwise independent hash family; this message 
has length $O(\log n) = O(\sqrt{m})$. With high probability,
$h$ is injective on the set
$\{j : x_j = 1\}$. The parties then run the online MA communication protocol of Theorem \ref{thm:pointquery} on the inputs $h(\bx)$ and $h(\iota)$
and output the result. The total cost of this protocol is $\tilde{O}(\sqrt{m})$ as claimed. In Appendix \ref{app:AMA}, we in fact show that up to logarithmic factors in $m$,
this online AMA protocol is optimal.

Meanwhile, the lower bound of Theorem \ref{thm:lb}
implies that the online MA communication complexity of this problem is $\Omega(m^{3/4})$. Indeed, if we have a protocol $\cP$ with $\hcost(\cP)=m^{3/4} > \log n$,  Theorem \ref{thm:lb}
implies that $\hcost(\cP)\vcost(\cP) = \Omega(m \log(n/m)) = \Omega(m^{3/2})$, and hence $\vcost(\cP) > m^{3/4}$.

To our
knowledge, this is the first such separation between online AMA and online MA communication complexity (we remark that polynomial separations between online MA and MAMA communication
complexity were already known, for problems including \idx\ and \disj\ \cite{AW, icalp}).
Indeed, all previous lower bound methods that apply to online MA communication complexity, such as the proof of \cite[Theorem 3.1]{icalp} and the methods of Klauck and Prakash \cite{klauck},
in fact yield equivalent AMA lower bounds. At a high level, the reason is that these methods work via round reduction -- they remove the need for Merlin's message. 
They therefore turn any online MA protocol for a function $\fn$ into
an online ``A'' protocol for $\fn$, which is really just a one-way randomized protocol without a prover, allowing one to invoke a known lower bound on
the one-way randomized communication complexity of $\fn$.
Similarly, they turn an online AMA protocol for $\fn$ into an online AA protocol, which is also just a one-way randomized protocol for $\fn$. 

The reason Theorem \ref{thm:lb} is capable of separating online AMA from MA communication complexity is that the reduction in the proof of Theorem \ref{thm:lb}
turns an online MA protocol for the $(m, n)$-Sparse \idx\ Problem into an online MA protocol for the (dense) \idx\ Problem with related costs. However, 
the natural variant of the reduction applied to an online AMA protocol for the  $(m, n)$-Sparse \idx\ Problem yields an online MAMA protocol for the dense \idx\ Problem, \emph{not} an online AMA protocol (see Appendix \ref{app:AMA} for details).
And the dense \idx\ Problem has an online MAMA protocol that is polynomially more efficient than any online AMA protocol (see e.g. \cite{AW, vldb}).

\section{Prescient \aSchemes\ for Sparse Disjointness and Frequency Moments} \label{sec:disj}

In this section and the next, we describe \aschemes\ for the
$\s$-Disjointness (\kdisj) and Frequency Moment ($F_k$) problems.
These \aschemes\ contain the main ideas of the paper. 
%We informally summarize them below.

\begin{table}[tb]
\centering
\begin{tabular}{r@{}l|c|c|c}
%\hline
\multicolumn{2}{c|}{
 Scheme\ Costs} & Completeness & Prescience & Source \\ \hline
 $(\s \log\s)^{2/3},~ (\s \log\s)^{2/3})$&: $\s = \Omega(\log n)$ & Perfect & Prescient & Theorem \ref{thm:pdisjoint}\\
 $(\anncost \log n,~ c_v \log n)$&: $\anncost c_v \geq n$ & Perfect & Online & \cite{icalp}\\
%\hline
$(\s\log n,~ \log n)$ && Perfect & Online &  \cite{icalp}\\
%\hline
%\hline
$(\anncost \log n\, \log_{c_v} \s,~ c_v \log n\, \log_v \s)$&: $\anncost=\s c_v^{-1/2}$	& Imperfect & Online & Theorem \ref{thm:of2}\\
%\hline
\end{tabular}
\caption{Comparison of our \kdisj\ \aschemes\ to prior work. Ours
  are the first \aschemes\ to achieve annotation length and space usage 
that are both sublinear in $\s$ for $\s \ll \sqrt{n}$, and we strictly improve over the MA communication cost (online or prescient) of prior \aschemes\ whenever $\s=o(n)$.}
\label{tabdisj}
%\vspace{-0.15in}
\end{table}

\subsection{Background: Optimal \aSchemes\ for Dense Problems}

We begin with a \ascheme\ achieving optimal tradeoffs between annotation
length and space usage for a broad class of dense problems.  Though this
\ascheme\ follows readily from prior work \cite{icalp, esa}, we describe it
in detail for completeness. This \ascheme is a good example of a {\em sum-check
scheme} as described in Section~\ref{sec:related}, and is based on the
Aaronson--Wigderson MA protocol for $\disj$~\cite{AW}.

\begin{proposition} \label{prop:dense} 
  Let $f^{(1)}, \dots, f^{(\ell)}$ denote the frequency vectors of $\ell$ data
  streams, each over the universe $[n]$.  Let $g$ be an $\ell$-variate
  polynomial of total degree $d$ over the integers. Let $F=\sum_{i=1}^n
  g(f^{(1)}_i, \dots, f^{(\ell)}_i)$, and let $o$ be an a priori upper bound
  on $|F|$.  Then for positive integers $\anncost, c_v$ with $\anncost c_v
  \geq n$, there is an online $(d \anncost (\log n + \log o),\, \ell c_v (\log
  n + \log o))$-\ascheme for computing $F$ in the non-strict turnstile update
  model.
\end{proposition}
\begin{proof}
  We work on $\F_q$, the finite field with $q$ elements, for a suitably large
  prime $q$; the choice $q > 2 d(n+o)^2$ suffices.  $V$ treats each
 % \mnote{Justin: Is $q$ really large enough? Magnitude of $f$ entries? Partial sums?  Final error prob?}
  $n$-dimensional vector $f^{(j)}$ as a $\anncost \times c_v$ array with
  entries in $\F_q$, using any canonical bijection between $[\anncost]
  \times [c_v]$ and $[n]$, and interpreting integers as elements of $\F_q$ in
  the natural way.  Through interpolation, this defines a unique bivariate
  polynomial $\tilde{f}^{(j)}(X, Y) \in \F_q[X,Y]$ of degree $\anncost-1$ in $X$
  and $c_v-1$ in $Y$, such that for all $x \in [\anncost]$, $y \in [c_v]$,
  $\tilde{f}^{(j)}(x,y) = f^{(j)}(x,y)$. 

  The polynomials $\tilde{f}^{(j)}$ can then be evaluated at locations outside
  $[\anncost] \times [c_v]$, so in the \ascheme\ $V$ picks a random position
  $r \in \F_q$, and evaluates $f^{(j)}(r,y)$ for all $j \in [\ell]$ and $y \in
  [c_v]$; $V$ can do this using $c_v$ words of memory per vector $f^{(j)}$ in
  a streaming manner~\cite[Theorem 4.1]{icalp}.  Let $\tilde{g}$ denote the
  total-degree-$d$ polynomial over $\F_q$ that agrees with $g$ at all inputs
  in $\F_q^\ell$.  $P$ then presents a polynomial $b(X)$ of degree at most
  $d(\anncost - 1)$ that is claimed to be identical to $\sum_{y \in [c_v]}
  \tilde{g}(\tilde{f}^{(1)}(X, y), \dots, \tilde{f}^{(\ell)}(X, y))$. 

  $V$ checks that $b(r) = \sum_{y \in [c_v]} \tilde{g}\left(\tilde{f}^{(1)}(r,
  y), \dots, \tilde{f}^{(\ell)}(r, y)\right)$. If this {\em sum check} passes,
  then $V$ believes $P$'s claim and accepts $\sum_{x \in [\anncost]} b(x)$ as
  the correct answer.  It is evident that this \ascheme satisfies perfect
  completeness. The proof of soundness follows from the Schwartz-Zippel lemma:
  if $P$'s claim is false, then
  \[
    \Pr\bigg[b(r) = \sum_{y \in [c_v]} \tilde{g}
      \left(\tilde{f}^{(1)}(r, y), \ldots, \tilde{f}^{(\ell)}(r, y)\right) 
    \bigg]
    \leq d (\anncost-1)/q \, . \qedhere
  \]
\end{proof} 

\subsection{A Prescient \aScheme\ for Sparse Disjointness}

An important special case of the communication problem \disj is when Alice's
and Bob's input sets are promised to be small, i.e., have size at most $\s \ll
n$. These should be thought of as {\em sparse} instances. The sparsity parameter
$\s$ has typically been denoted by the letter $k$ in the communication complexity literature, and the problem
has typically been referred to as $k$-\disj\ rather than \kdisj; we
use $\s$ rather than $k$ for consistency with our notation in the rest of the paper (where $\s$ denotes the sparsity of a data stream).

Among the original
motivations for studying this variant is its relation to the
clique-vs.-independent-set problem introduced by
Yannakakis~\cite{Yannakakis91} to study linear programming formulations for
combinatorial optimization problems. More recent motivations include
connections to property testing~\cite{BlaisBM12}. A clever protocol of
H{\aa}stad and Wigderson~\cite{HastadW07} gives an optimal $O(\s)$
communication protocol for \kdisj, improving upon the trivial $O(\s\log n)$ 
and the easy $O(\s\log \s)$ bounds.
This protocol requires considerable interaction between Alice and Bob, a
feature that turns out to be necessary. Recent results of Buhrman
\etal~\cite{BuhrmanGSMdW12} and Dasgupta \etal~\cite{DasguptaKS12} give tight
$\Theta(\s\log \s)$ bounds for \kdisj\ in the one-way model. Very recently,
Brody \etal~\cite{BrodyCK12} and Sa\u{g}lam and Tardos~\cite{SaglamT12-priv}
have given tight rounds-vs.-communication tradeoffs for \kdisj.

Here we obtain the first nontrivial bounds for \kdisj\ in the annotated
streams model, and thus also in the online MA communication model.

\begin{definition} \label{def:disj} 
In the \kdisj\ problem, the data stream specifies two multi-sets $S, T \subseteq [n]$, with $\|S\|_0,\|T\|_0\leq \s$, where $\|S\|_0$
denotes the number of distinct items in $S$.
An update of the form $((0, i), \delta)$ is interpreted as an insertion of $\delta$ copies of item $i$ into set $S$, and an update of the form $((1, i), \delta)$
is interpreted as an insertion of $\delta$ copies of item $i$ into $T$.
The goal is to determine whether or
not $S$ and $T$ are disjoint. 
\end{definition}

Notice Definition \ref{def:disj} allows $S$ and $T$ to be multi-sets, but assumes
%Notice that as $S$ and $T$ are multi-sets, Definition \ref{def:disj}
%only makes sense in 
the strict turnstile update model, where the frequency of each item is non-negative.

\begin{theorem} \label{thm:pdisjoint}
Assume $\s > \log n$. There is a \pscheme{( (\s \log \s)^{2/3}, (\s \log \s)^{2/3})} for \kdisj\ with perfect completeness
in the strict turnstile update model. In particular, the MA-communication complexity of \kdisj\ is $O( (\s \log \s)^{2/3})$.
Any prescient $(\anncost, c_v)$ protocol requires $\anncost c_v = \Omega(\s)$.
\end{theorem}

\begin{proof}
Obviously if $S$ and $T$ are not disjoint, the prescient prover can provide an item $i \in S \cap T$ at the start of the stream
and the verifier can check that $i$ indeed appears in both $S$ and $T$. The total space usage and annotation length is just $O(\log n)$ 
in this case.

Suppose now that $S$ and $T$ are disjoint. We first recall that a
$(\sqrt{n}\log n, \sqrt{n}\log n)$-\ascheme\ for \disj\ follows from 
Proposition \ref{prop:dense}, with $f^{(1)}$ and $f^{(2)}$ set to the indicator vectors of $S$ and $T$ respectively, and $g$ equal to the product function.
%This \ascheme\ is based on the Aaronson-Wigderson MA-communication protocol for \disj implied by Proposition \ref{prop:dense}.
We refer to this as the dense \disj \ascheme\ because its cost does not improve if $|S|$ and $|T|$ are both $o(n)$.

Our prescient \ascheme\ for \kdisj\ works as follows. At the start of the stream, the prover describes a hash function
$h:[n]\to[r]$, for some smaller universe $[r]$, with the property
that $h$ is injective on $S \cup T$. 
We will write $h(S)$ to denote the result of applying $h$ to every
member of $S$. 
The parties can now run the 
%Aaronson-Wigderson based 
dense \disj\ \ascheme\ whereby $P$ convinces 
$V$ that $h(S)$ and $h(T)$ are disjoint. Given the existence of an 
injective function $h$,
perfect completeness follows from the fact that if $S$ and $T$ are disjoint, so are
$h(S)$ and $h(T)$, combined with the perfect completeness of the dense \disj\ \ascheme.
Soundness follows from the fact that if $i \in S \cap T$, then $h(i) \in h(S) \cap h(T)$ i.e. if
$S$ and $T$ are not disjoint, then the same holds trivially for $h(S)$ and $h(T)$.

The dense \disj\ \ascheme\ run on $h(S)$ and $h(T)$ requires annotation
length and space usage $O(\sqrt{r} \log r)$. 
We now show that, for a suitable choice of $r$, $P$'s description of $h$
is also limited to $O(\sqrt{r} \log r)$ communication, balancing out the
cost of the rest of the \ascheme.

A family of functions $\mathcal{F} \ceq [r]^{[n]}$ is said to be
$\kappa$-perfect if, for all $S \ceq [n]$ with $|S| \le \kappa$, there
exists a function $h \in\mathcal{F}$ that is injective when restricted to $S$.
Fredman and Koml\'os~\cite{FredmanK84} have shown that for
all $n \geq r \geq \kappa$, there exists a $\kappa$-perfect family
$\mathcal{F}$, with
\[
  |\cF| ~\le~ (1 + o(1)) \left(
    \frac{\kappa \log n}{-\log(1-t(r,\kappa))}\right) \, ,
\]
where
\[
  t(r,\kappa) ~:=~ \prod_{j=1}^{\kappa-1} \left(
    1 - \frac{j}{r}\right) \, .
\]
For $r \ge 2\kappa$, we can use the crude approximation
\[
  -\log(1-t(r,\kappa))
  ~\ge~ t(r,\kappa)
  ~\ge~ \left(1 - \frac{\kappa}{r}\right)^\kappa
  ~\ge~ e^{-2\kappa^2/r}
\]
to obtain the bound $|\cF| = O(\kappa e^{2\kappa^2/r} \log n)$, which implies
\[
  \log|\cF| ~=~ O(\kappa^2/r) \, ,
\]
for $\kappa^2/r = \Omega(\log \kappa)$ and $\kappa = \Omega(\log n)$.

Let us pick a family $\cF$ that is $(2\s)$-perfect.  Once $P$ and $V$
agree upon such a family $\cF$, the prover, upon seeing the input
sets $S$ and $T$, can pick $h \in\cF$ that is injective on $S\cup T$.
Describing $h$ requires $O(\s^2/r)$ bits; $P$ sends this to $V$ before the stream
is seen, and $V$ stores it while observing the stream in order to run
the dense
%Aaronson-Wigderson-based 
\disj\ \ascheme\ on $h(S)$ and $h(T)$.  
To balance out this
communication with the $O(\sqrt{r} \log r)$ cost of running the
%Aaronson-Wigderson-based 
dense \disj\ \ascheme\ 
%of \cite[Theorem 4.1]{icalp} 
on $h(S)$ and $h(T)$, we choose $r$ so that
\[
  \frac{\s^2}{r} ~=~ \Theta(\sqrt{r} \log r) \, .
\]
This is achieved by setting $r = \s^{4/3}/\log^{2/3} \s$. The resulting
upper bound is that both the annotation length and verifier's space usage are $O\left((\s\log \s)^{2/3} \right).$

The lower bound follows from known lower bounds for dense streams \cite{icalp}.
\end{proof}

\subsection{A Prescient \aScheme\ for Frequency Moments }
\label{sec:pf2}

\begin{table}[tb]
\centering
\begin{tabular}{r@{}l|c|c|c}
%\hline
\multicolumn{2}{c|}{ \aScheme\ Costs} & Completeness & Prescience & Source \\ \hline
 $(k^2 \anncost \log n,~ kc_v\log n)$&: $\anncost c_v \geq n$ & Perfect & Online & \cite{icalp}\\
%\hline
$(\s \log n,~ \log n)$ && Perfect & Online &  \cite{icalp}\\
%\hline
 $(k^2\s^{2/3} \log n,~ k\s^{2/3} \log n)$&& Perfect & Prescient & Theorem \ref{thm:pf2}\\
%\hline
$(k^2 \s\, c_v^{-1/2} \log n\, \log_{c_v} \s,~ k c_v \log n\, \log_{c_v}\s)$&: $c_v > 1$ & Imperfect & Online & Theorem \ref{thm:of2}\\
%\hline
\end{tabular}
%\end{center}
\caption{Comparison of our \fk\ \aschemes\ to prior work. Ours are the first \aschemes\ to achieve annotation length and space usage 
that are both sublinear in $\s$ for $\s \ll \sqrt{n}$, and we strictly improve over the MA communication cost of prior protocols (online or prescient) whenever $\s=o(n)$.}
\label{tabkmoment}
\end{table}

We now present prescient \aschemes\ for the $k$th Frequency
Moment problem, \fk. 
\begin{definition} 
\label{def:fk} In the $F_k$ problem, the data stream $\bx$ consists of a sequence of updates of the form $(i, \delta)$, and the frequency 
of item $i$ is defined to be $f_{i} = \sum_{(j_\ell, \delta_\ell) \in \bx:
  j_\ell = i} \delta_\ell$. 
The goal is to compute $\fk = \sum_{i \in [n]} f_i^k.$
\end{definition}

The idea behind the \ascheme, as in the case of \kdisj, is that $P$ is supposed to specify a ``hash function'' $h$ to reduce the universe size in
a way that does not introduce false collisions. However, for 
\fk\
%the $k$th Frequency Moment problem, 
it is essential that $V$ ensure $h$
is truly injective on the items appearing in the data stream. This is in contrast to 
\kdisj, where a weaker notion than injectiveness was sufficient to guarantee soundness.
The fundamental difference between the two problems is that for \kdisj, collisions only ``hurt the prover's claim'' that the two sets
are disjoint, whereas for \fk\ the prover could try to use collisions to convince the verifier that the answer to the query
is higher or lower than the true answer.

%Our collision-detection \ascheme\ is the reason that the costs stated in Theorems \ref{thm:pf2} and \ref{thm:of2} depend on the stream length $m$
%instead of the stream sparsity $M$.

\begin{theorem} \label{thm:pf2}
There is a \pscheme{(k^2\s^{2/3} \log n, k\s^{2/3} \log n)}
for computing \fk\ over a data stream of sparsity $\s$ in the strict turnstile update model.
This \ascheme\ has perfect completeness. Any prescient $(\anncost, c_v)$ protocol requires $\anncost c_v = \Omega(\s)$.
\end{theorem}

\begin{proof}
The idea is to have the prover specify for the verifier
a perfect hash function $h: [n] \rightarrow [r]$, where $r$ is to be determined later, i.e.
$P$ specifies a hash function $h$ such that for all $x \neq y$ appearing in at least one update in the data stream,
$h(x) \neq h(y)$. The verifier stores the description of $h$, and while observing the stream
 runs the dense $F_k$ \ascheme\ of Proposition \ref{prop:dense} on the derived stream in which each update $(i, \delta)$ is replaced with the update $(h(i), \delta)$.

As discussed above, it is essential that $V$ ensure $h$ is injective
on the set of items that have non-zero frequency, as otherwise $P$ could try to introduce collisions to try to trick the verifier. 
To deal with this, 
we introduce a mechanism by which $V$ can ``detect'' collisions.

\begin{definition}
\label{def:inject}
Define the problem \inject\ as follows.
We observe a stream of tuples $t_i = ( (x_i, b_i), \delta_i)$.
Each $t_i$ indicates that $\delta_i$ copies of item $x_i$ are placed in
bucket $b_i \in [r]$.
We allow $\delta_i$ to be negative, modeling deletions, and
refer to the quantity $f_{(j, b)} = \sum_{i: (x_i, b_i) = (j, b)} \delta_i$
as the \emph{count} of pair $(j, b)$. 
We assume the strict turnstile model, so that 
for all pairs $(j, b)$ we have $f_{(j, b)} \geq 0$.

We say that the stream is an injection if for every two pairs $(j, b)$ and $(j', b)$ with positive counts, it holds that $j=j'$. 
Define the output as 1 if the stream defines an injection, and 0 otherwise.  
\end{definition}

\begin{lemma}
\label{lemma:injection}
For any $\anncost c_v \geq r$, there is an \oscheme{(\anncost \log r,
  c_v\log r)} for determining whether a stream in the strict turnstile
model is an injection.
\end{lemma}
\begin{proof}
Say that bucket $b$ is \emph{pure} if there is at most one $j \in [n]$ such that $f_{(j, b)}> 0$.
The stream defines an injection if and only if every bucket $b$ is pure. 

Notice that a bucket $b$ is pure if and only if the variance of the item identifiers mapping to the bucket with positive count is zero. 
Intuitively, our \ascheme\ will compute the sum of the these variances across all buckets $b$; this sum will be zero if and only if the stream defines
an injection. Details follow.

Define three $r$-dimensional vectors $u, v, w$ as follows:

$$u_b = \sum_{j \in [n]} f_{(j, b)},$$
$$v_b = \sum_{j \in [n]} f_{(j, b)}j,$$
$$w_b = \sum_{j \in [n]} f_{(j, b)} j^2.$$

It is easy to see that if bucket $b$ is pure then $v_b^2 = u_b \cdot w_b$.
Moreover, if bucket $b$ is impure then $v_b^2 < u_b w_b$;
this holds by the Cauchy-Schwarz inequality applied to the $n$-dimensional vectors whose $j$th entries are $\sqrt{f_{(j, b)}}$ and $\sqrt{f_{(j, b)}} \cdot j$ respectively (the strict inequality holds because for an impure bucket $b$, the vector
given by $\sqrt{f_{(j, b)}} \cdot j$ is not a scalar multiple of the vector given by $\sqrt{f_{(j, b)}}$). 
Here, we are exploiting the assumption that $f_{(j, b)} \geq 0$ for all pairs $(i, b)$, as this allows us to conclude that all $\sqrt{f_{(j, b)}}$ values are 
real numbers. 

It follows that $\sum_{b \in [r]} v_b^2 = \sum_{b \in [r]} u_b \cdot w_b$ if and only if the stream defined an injection. Both quantities can be computed 
using the ``dense'' \ascheme\ of Proposition \ref{prop:dense}.
Notice that each update $t_i = ((x_i, b_i), \delta_i)$ contributes independently to each of the vectors $u$, $v$, and $w$, and hence 
it is possible for $V$ to run the \ascheme\ of Proposition  \ref{prop:dense} on these vectors as required. 
This yields an online $(\anncost \log r,c_v \log r)$-\ascheme\ for the injection problem for any $\anncost c_v \geq r$ as claimed.
%$P$ and $V$ will run a number of \aschemes\ in parallel to verify this.
%\begin{enumerate}
%\item They run an $F_2$ \ascheme\ on the $r$ buckets, where the weight
  %in the buckets is the sum of the $x_i$'s mapped into that
 % bucket. That is, define a vector $y$ such that $y_b = \sum_{i: b_i = b}
 % x_i$, and compute $F_2(y) = \sum_{b=1}^{r} y_b^2$. 
%
%\item They run a dot-product \ascheme\ on the $r$ buckets, for two vectors.
%Vector $c_v$ is defined to count the number of items in each bucket,
%i.e. $v_b = \sum_{i: b_i = b} 1$. 
%Vector $w$ is defined to sum the squared value of items mapped into each bucket, i.e. $w_b = \sum_{i: b_i = b} x_i^2$.
%\end{enumerate}
%Running both \aschemes\ in parallel, $V$ accepts if $(v \cdot w) = F_2(y)$.
%To understand this compound \ascheme\, consider a single bucket $b$.
%We have $y_b^2 = v_b w_b$ if and only if there is a unique item mapped
%to this bucket.  
%This follows from the fact that the 
%sample variance of the multiset of item identifiers mapped to bucket b
%is zero if and only if all these identifiers are the same.  
%Otherwise, $y_b^2 < v_b \cdot w_b$.
%Summing over buckets, 
%we have $v \cdot w = y \cdot y$ if and only if all buckets receive a
%unique item (or are empty); else, $y \cdot y < v \cdot w$.
\end{proof}

Returning to our \fk\ \ascheme, $P$ specifies a hash function $h$ claimed to be one-to-one on the set of items
that appear in one or more updates of the stream $\bx$. $V$ verifies that $h$ is injective using the \ascheme\ of Lemma~\ref{lemma:injection}.
 If this claim is true, then $\fk(\mathcal{\bx})=\fk(h(\bx))$, the frequency moment of the mapped-down stream, and $P$ can prove this by running
the \ascheme\ of \cite[Theorem 4.1]{icalp} on the derived stream $h(\bx)$. 

Perfect completeness follows from $P$'s ability to find a perfect hash
function just as in Theorem \ref{thm:pdisjoint}. %\footnote{There, we
 % described the hash function as just separating $S$ and $T$ rather
 % than injective on $S \cup T$.  
%However, we can construct a perfect hash function in the same manner without affecting the asymptotic description length.}, followed by the perfect completeness of
%the \fk\ \ascheme\ of
%\cite[Theorem 4.1]{icalp}.
%
Soundness follows from the soundness of the \inject\ \ascheme\ of Lemma \ref{lemma:injection}, in addition to the soundness
property of the $F_k$ \ascheme\ of \cite[Theorem 4.1]{icalp}. 

To analyze the costs, note that by using the hash family of Fredman and Koml\'{o}s \cite{FredmanK84}, the annotation length and space cost due to specifying and storing the hash function $h$ is $O(\s^2\log n/r)$.
The annotation length and space cost of the dense $F_k$ \ascheme\ of Proposition \ref{prop:dense} are $O(k^2 \anncost \log r)$ and $O(kc_v \log r)$ for any $\anncost c_v \geq r$.
The annotation length and space cost of the \inject\ \ascheme\ can be set to $O(\anncost \log r)$ and $O(c_v \log r)$ respectively.
Setting $r=\s^{4/3}$ and $\anncost =  c_v = \s^{2/3}$ yields the desired costs.
\end{proof}

\eat{One point to remark upon is that the \inject\ \ascheme\ assumes that
the associated frequencies of its input stream are non-negative: it
may give erroneous answers otherwise. 
However, we are able to apply our \ascheme\ to input streams which
contain negative values, since the input to \inject\ does {\em not}
contain negative values. 
%To clarify, we do not know how to modify to handle
%deletions in the stream of updates $(x_i, b_i)$. Fortunately, we do
%not run the \inject\ protocol on the raw stream who's frequency moment
%we want to compute, but only on a derived stream, so there is no issue
%when dealing with deletions in the raw stream. 
That is, within the \inject\ \ascheme, we treat deletions of items in
the input stream as insertions of items into the derived stream of
$(x_i, b_i)$ updates. 
This does not affect the frequency moment calculation, because the
\inject\ \ascheme\ is used to check the property that each bucket is
``pure''. 
However,  this approach of treating deletions as insertions within the
\inject\ \ascheme\ is the reason that our costs are proportional to the
stream length $\length$ rather than the stream sparsity $\s$.
}

\section{An Online \aScheme\ for Frequency Moments} \label{sec:of2}

We now give an online version of \fk\ \ascheme\ of Theorem \ref{thm:pf2}. A simple modification of this \ascheme\ yields the \ascheme\
for \kdisj\ with analogous costs as claimed in Row 4 of Table \ref{tabdisj}. 
In addition to avoiding the use of prescience, our online \ascheme\
avoids requiring $V$ to explicitly store the hash function sent by
$P$, allowing us to achieve a much wider range of tradeoffs between
annotation size and space usage relative to Theorems \ref{thm:pdisjoint} and \ref{thm:pf2}.

%This \ascheme\ is more expensive than
%the online $k$-\disj\ \ascheme\ of Theorem \ref{thm:odisj} because for
%the \fk\ \ascheme\
%it is harder for $V$ to ensure that $P$ is applying a universe-reduction function $g^*$ consistently
%i.e. that $P$ does not try to get $V$ to use different universe-reduction mappings on different parts of the stream (this alludes to Idea 2 in Section \ref{sec:informal}). In the $k$-\disj\ \ascheme, $V$ could enforce a form of consistency that was sufficient for the problem
%by ensuring that for each $x$ there was at least one run of the
%\ascheme\ for which $x$ was not marked as an exception. 
%Here, repetition is insufficient, and we resort to running multiple
%parallel instances of the
%\pq\ \ascheme\ to ensure consistency, which is slightly costlier.
%We also make further use of the \inject\ collision detection \ascheme,
%noting that it does not require prescience. 
%This collision detection protocol is online,
%which does not matter in the context of the prescient \fk\ protocol we are currently constructing, but
%will be important in Theorem \ref{thm:of2} when were construct an online protocol for the same problem.

\begin{theorem} \label{thm:of2}
For any $c_v > 1$, there is an online $(k^2 \s c_v^{-1/2} \log n\, \log_{c_v} \s,~ kc_v \log n\, \log_{c_v} \s)$-\ascheme\ for 
\fk\ in the strict turnstile model
for a stream of sparsity $\s$ over a universe of size $n$. Any online $(\anncost, c_v)$-scheme for this problem with $\anncost \geq \log n$ 
requires $\anncost c_v = \Omega(\s \log(n/\s))$.
%For any $1 < v < M$, there is an \oscheme{(k^2h\log n, kv\log n)}
%for computing \fk\ over
%a stream of length $m$ over a universe of size $n$ in the non-strict turnstile update model. This \ascheme\ satisfies imperfect completeness.  %This is not quite as good as the$k$-Disjointness protocol, where we have an $(h, v)$ protocol for $\anncost c_v = m\sqrt{v}$ which is always less than $m^{3/2}$ if $v < m$.
\end{theorem}

Notice that the annotation length is less than $\s \log n$ for any $c_v = \s^{\Omega(1)}$, and therefore
this protocol is not subsumed by the simple ``sparse'' \ascheme\
(second row of Table \ref{tabkmoment}) in which $P$ just
replays the entire stream in a sorted order, and $V$ checks this is done correctly using fingerprints.
Notice also that the product of the space usage and annotation length is
$k^3 \s c_v^{1/2} \log^2 n \log^2_{c_v} \s$, which is in $o(n)$ for many interesting parameter settings.
This improves upon the dense sum-check \ascheme
(first row of Table \ref{tabkmoment}) in such cases.

%Notice both $\anncost$ and $c_v$ can be sublinear in $m$, which is our baseline
%for a non-trivial \ascheme, but this is only non-trivial if 
%$v = \omega(m^{1/2})$ (else $P$ just presents the stream in sorted
%order, and $V$ uses fingerprints to ensure that $P$ does so correctly)
%and $h=\omega(m^{1/2})$ (else $V$ can just store the whole input stream).

\subsection{An Overview of the \aScheme}

In order to achieve an online \ascheme,
we examine how to construct perfect hash functions such as those used in the prescient \fk\ \ascheme\ of Theorem \ref{thm:pf2}.
 Let $S$ be the set of $\s$ items with non-zero frequency at the end of the stream: we
want the hash function to be one-to-one on $S$. Choose a hash
function $h$ at random from pairwise independent hash family mapping
$[n]$ to $[r]$, for $r$ to be specified later -- this requires just
$O(\log n)$ bits to specify. We only expect $O(\s^2/r)$ pairs to
collide under $h$, which means that with constant probability there will be
$O(\s^2/r)$ collisions if $h$ is chosen as specified. 
The final hash function $h^*$ is specified by writing
down $h$ (which takes only $O(\log n)$ bits), followed by the items
involved in a collision and some special locations for them. 
The total (expected) bit length to specify this hash function is $O(\s^2 \log (n)/r)$.  

In our online \fk\ \ascheme, $P$ will send such an $h$ at the start of the stream. Notice $h$ does not depend on the stream itself -- it is just a random pairwise independent hash function -- so $P$ is not using prescience.
$P$ also has no incentive not to choose $h$ at random from a pairwise independent hash family, since the only purpose of choosing $h$ in this manner is to minimize the number of collisions
under $h$. If $P$ chooses $h$ in a different way, $P$ simply risks that there are too many collisions under $h$, causing $V$ to reject.

Now while $V$ observes the stream, she runs the online
sum-check scheme for $F_k$ given in Proposition
\ref{prop:dense} on the mapped-down universe of size $r$, using $h$ as the mapping-down function. 
At the end of the stream, $P$ is asked to retroactively specify a hash function $h^*$ that is one-to-one on $S$ as follows.
$P$ provides a list $L_0$ of all items in $S$ that were involved in a collision under $h$, accompanied by their frequencies.
Assuming that these items and their frequencies are honestly specified by $P$, $V$ can compute their contribution to \fk\ and \emph{remove them} from the stream. By design,
$h^*$ is then (claimed to be) 
injective on the remaining items. $V$ can confirm this
tentatively using the \inject\ \ascheme\ of Lemma \ref{lemma:injection}. 

The remainder of the \ascheme\ is devoted to making the correctness a
certainty by ensuring that the items in $L_0$ and their frequencies are as claimed (we stress that while our exposition of the \ascheme\ is modular,
all parts of the \ascheme\ are executed in parallel, with no communication ever occurring from $V$ to $P$).
A naive approach to checking the frequencies of the items in $L_0$ would be to run $|L_0|$ independent \pq\ \aschemes, one for each item in $L$; however there are too many items in $L_0$ for this to
be cost-effective. Instead, we check all of the frequencies as a batch, with a (sub-)\ascheme\ whose cost is roughly equal to that of a single \inject\ query.

This (sub-)\ascheme\ can be understood as proceeding in stages, with
each stage $i$ using a different pairwise independent hash function
$h_i$ to map down the full original input. Say that
an item $j$ is \emph{isolated} by $h_i$ if $j$ is not involved in a
collision under $h_i$ with any other item with non-zero frequency in the original data stream $\bx$.
The goal of stage $i$ is to isolate a large fraction of items which were not isolated by any previous stage.

A key technical insight is that at each stage $i$, it is possible for $V$ to ``ignore'' all items that are not isolated at that stage.
This enables $V$ to check that the frequencies of all items 
that \emph{are} isolated at stage $i$ are as claimed.
We bound the number of stages that are required to isolate all items if $P$ behaves as prescribed -- if $P$ reaches an excessive
number of stages, then $V$ will simply reject. 

%Each stage $i$ begins with a list $L_{i-1}$ of items that were not isolated in any previous stage.
%At the end of the stream, $P$ must retroactively modify each $h_i$ by
%specifying a list $L_{i}$ of items that suffer collisions under \emph{all} hash functions $h_j$ for $j \leq i$.
%The underlying goal is to use these stages to successively whittle
%down the number of colliding items, so that the the list $L_{i^*}$ is empty in the final stage $i^*$.
%At this point, $P$ need not specify any modifications to the hash function. 

%The intuition for why this \ascheme\ is sound is as follows. 
%Suppose item $j$ is (claimed to be) isolated at Stage $i_1$. 
%If the claimed frequency of item $j$ is false, then $P$ is left with a
%false claim about which items are isolated at Stage $i_1+1$. 
%To convince $V$ of this claim, $P$ would have to lie about the
%frequency of some items (claimed to be isolated) at a later stage
%$i_2$, but this would leave $P$ with a false claim to prove at round
%$i_2 + 1$. 
%This line of argument continues until we reach the final stage $i^*$, 
%And so on until we arrive at the final stage, 
%at which point $P$ no longer claims any collisions, 
%and we can therefore argue that $P$ has no means to convince $V$ of
%any false claims. 

\subsection{Details of the \aScheme}

\begin{proofof}[Proof of Theorem \ref{thm:of2}:]
Let $r=\s c_v^{1/2}$. $P$ sends a hash function $h: [n] \rightarrow [r]$ at the start of the stream, claimed to be chosen at random from a pairwise independent hash family. 
While observing the stream, $V$ runs the dense online
sum-check scheme for $F_k$ given in Proposition
\ref{prop:dense} on the mapped-down universe $[r]$. 
Let $S$ be the set of items with non-zero frequency at the end of the stream. After the stream is observed, 
$P$ is asked to provide a list $L_0$ of all items with nonzero frequency that were involved in a collision, followed by a claimed frequency $f_i^*$ for each $i \in L_0$.

Assuming that these items and their frequencies are honestly specified
in $L_0$ by $P$, $V$ can compute their contribution $C_0=\sum_{i \in L_0}
f_i^*$ to \fk\ and then remove them from the stream by processing
updates $U=\{(i, -f_i^*): i \in L_0\}$ within the dense $F_k$
\ascheme. $h$ is injective on the remaining items.
$V$ can confirm this using the \inject\ \ascheme\ of Lemma
\ref{lemma:injection} (conditioned on the assumed correctness of $L_0$). 
Thus the dense $F_k$ \ascheme\ will output $C_1 = \sum_{i \not\in L_0} f_i^k$. Assuming all of $V$'s checks within the dense $F_k$ \ascheme\ pass, $V$ outputs $C_0 + C_1$ as the answer.
  
The remainder of the \ascheme\ is directed towards
determining that the frequency of items in $L_0$ are correctly
reported. We abstract this goal as the following problem.

\begin{definition} Define the $\ell$-\multiindex\ problem as follows. Consider a data stream $\bx \circ L_0$, where $\circ$ denotes concatenation. 
$\bx$ is a usual data stream in the strict turnstile model, while $L_0$ is a list of $\ell$ pairs $(i, f_i^*)$. Let $f$ be the frequency vector of $\bx$. 
The desired output  is 1 if $f_i=f_i^*$ for all $i \in L_0$, and 0 otherwise.
\end{definition}

We defer our solution to the $\ell$-\multiindex\ problem
to Section~\ref{sec:multiindex}. For now, we state our main result
about the problem in the following lemma.

\begin{lemma} \label{lemma:multiindex} For all $c_v>1$, 
$\ell$-\multiindex has an online 
$(\s c_v^{-1/2} \log n\, \log_{c_v} \ell,~ c_v \log n \log_{c_v} \ell)$-scheme in the strict turnstile update model.
\end{lemma} 

\medskip \noindent
\textbf{Analysis of Costs.} 
Let $S$ be the set of items with non-zero frequency when the stream ends. First, we argue that if $r$ is the size of the mapped-down universe,
and $P$ chooses the hash function $h$ at random from a pairwise independent hash family, then with probability $9/10$, there will be
at most $10\s^2/r$  items in $S$ that collide under $g$. Indeed, by a union bound, the probability any item $i$ with non-zero count is involved in a collision is at most $\s/r$,
and hence by linearity of expectation, the expected number of items involved in a collision is at most $\s^2/r$.

So by Markov's inequality, with probability at least 9/10, the total
number of items involved in a collision will be at most $10\s^2/r=O(\s c_v^{-1/2})$ under the setting $r=\s c_v^{1/2}$.
Conditioned on this event, $P$ can specify the list $L_0$ and the associated frequencies with annotation length $O(\s c_v^{-1/2}\log n)$, and 
$V$ can use the \multiindex\ \ascheme\ of Lemma \ref{lemma:multiindex} with $\ell=O(\s c_v^{-1/2})$ to verify the frequencies of the items in $L_0$ are as claimed.
For any $c_v>1$, Lemma \ref{lemma:multiindex} under this setting of $\ell$ yields an $(\s c_v^{-1} \log n \cdot \log_{c_v} \ell,c_v \log n \cdot \log_{c_v} \ell)$-scheme.

%We set $V$'s space usage to be $v\log n$.
Running all of the sum-check \aschemes\ (i.e., the \inject\ \ascheme\ and the
\fk\ \ascheme\ itself) on the mapped-down universe requires annotation
$O(k^2r c_v^{-1} \log r)$ and space $O(kc_v\log r)$ for $V$; 
in total, this provides an \oscheme{(\s^2\log n/r + k^2r\log n/c_v + k\s c_v^{-1} \log n \cdot \log_{c_v} \s,c_v \log n \cdot \log_{c_v} M)}.

 Since we set $r=\s c_v^{1/2}$, we obtain a \oscheme{(k^2\s c_v^{-1/2}\log_{c_v}(\s), kc_v\log n\log_v(\s))} for
any $c_v > 1$.

The lower bound stated in Theorem \ref{thm:of2} follows from Theorem \ref{thm:lb} and an easy reduction from the $(\s, n)$-sparse \idx\ problem.
\end{proofof}

%\vspace{-7mm}

\subsection{A Scheme for MultiIndex: Proof of Lemma \ref{lemma:multiindex}}
\label{sec:multiindex}
Before presenting an efficient online \ascheme\  for the $\ell$-\multiindex\ Problem, we 
define two ``sub''problems, which apply a function to only a subset 
of the desired input. 

\begin{definition} 
Define the problem \subinj as follows. 
We observe a stream of tuples $t_i = (x_i, b_i, \delta_i)$,  followed by a vector $z \in \{0, 1\}^r$. 
As in the \inject\ problem, each $t_i$ indicates that $\delta_i$ copies of item $x_i$ are placed in
bucket $b_i \in [r]$.

We say that the stream defines a {\em subinjection}
based on $z$ if for every $b$ such that $z_b\geq 1$, for every two pairs $(x, b)$ and $(y, b)$ with positive counts, it holds that $x=y$. 
The \subinj problem is to decide whether the stream defines a
subinjection based on $z$.
%Define the output to be 1 if the stream defines a subinjection.
%That is, there is an injection if for any tuples $t_i$, $t_j$, we have that $b_i = b_j \Longrightarrow x_i = x_j.$
\end{definition}

\noindent Notice that the \inject\ problem is a special case of the \subinj\ problem with $z_i=1$ for all $i$.

\begin{lemma} \label{lemma:strict} For any $\anncost c_v \geq r$, there is an online
  $(\anncost \log r,c_v \log r)$-\ascheme\ for 
\subinj\ in the strict turnstile update model. Moreover, for any constant $c > 0$, this scheme can be instantiated to have soundness error $1/r^c$.
\end{lemma}

\begin{proof}
Define vectors $u$, $c_v$, and $w$ exactly as in the proof of Lemma \ref{lemma:injection}, 
and observe that the stream defines a sub-injection if and only if 
$\sum_{b \in [r]} z_b v_b^2 = \sum_{b \in [r]} z_b u_b w_b$. 
$V$ can compute both quantities using the dense \ascheme\ of Proposition \ref{prop:dense}, with the same asymptotic costs as the \ascheme\
of Lemma \ref{lemma:injection}. The soundness error can be made smaller than $1/r^c$ for any constant $c$ by running the \ascheme\ of Proposition \ref{prop:dense}
over a finite field of size $\poly(r)$, for a sufficiently fast-growing polynomial in $r$.
\end{proof}

We similarly define the problem \subftwo\ over a data universe of size $n$ based on a vector $z \in \{0, 1\}^n$ as $\sum_{i \in [n]} z_i f_i^2$,
the sum of squared frequencies of items indicated by $z$.
This too is a low-degree polynomial function of the input values, and
so Proposition \ref{prop:dense} implies \subftwo\ can be computed by an online 
$(\anncost \log r,c_v \log r)$-scheme in the general turnstile update model for any $\anncost,
c_v$ such that $\anncost c_v \geq r$ (and the soundness error in this protocol can be made smaller than $1/r^c$ for any desired constant $c$).

\medskip \noindent
\textbf{Online scheme for $\ell$-\multiindex.} The \ascheme\ can be thought of as proceeding in $t$ stages ($t$ will be specified later),
although these stages merely serve to partition the annotation: there
is no communication from $V$ to $P$ during these stages. 
Each stage $j$ makes use of a corresponding hash function
$h_j : [n] \rightarrow [r]$ for $r=\s c_v^{1/2}$. 
The $t$ hash functions are 
provided by $P$ at the start of the stream, so that $V$ has access
to them throughout the stream. 
Each $h_j$ is claimed to be chosen at random from a pairwise
independent hash family:  if they are, 
then there are unlikely to be too many collisions, 
so $P$ has no incentive not to choose $h_j$ at random.
Let $f$ denote the vector of frequencies defined by the input
stream, 
and let $f^{(0)}$ denote the vector satisfying $f^{(0)}_i=f_i$ for $i \in L_0$, and $f^{(0)}_i=0$ for $i \not\in L_0$.

Stage $j$ begins with a list $L_{j-1}$ of items. We will refer to these items as ``exceptions''.
$P$ provides a new list $L_j \subseteq L_{j-1}$ of items which
remain exceptions in stage $j$; $P$ implicitly claims that no items in $L_{j-1} \setminus L_j$ collide with
some other input items under hash function $h_j$. 
Let $z^{(j)}$ denote the indicator vector of the list of buckets
corresponding to $L_{j-1} \setminus L_j$, i.e. $z^{(j)}_{h_j(i)}=1$ if
$i \in L_{j-1} \setminus L_j$, and $z^{(j)}$ entries are $0$ otherwise. 
To check that no items in $L_{j-1} \setminus L_j$ collide under $h_j$, $V$ will use the  \subinj\ \ascheme\ 
 based on the indicator vector $z^{(j)}$ over the full original input $f$ as mapped by the hash function $h_j$. 
 Note that since the original input stream is in the strict turnstile update model, so is the stream on which the 
 \subinj\ \ascheme\ is run (as the \subinj\ \ascheme\ is simply run on the original input stream as mapped by the hash function $h_j$,
 based on the vector $z^{(j)}$).
 Note also that $L_{j-1}$ and $L_j$ are provided explicitly, so $V$ can
compute $z^{(j)}$ easily.\footnote{For example, $V$ can add one to the
  corresponding entry of $z^{(j)}$ for each item that is marked as an
  exception.  This will cause $z^{(j)}$ to count the number of exceptions
  in each bucket, rather than indicate them, but this does not affect
  the correctness.} 
  
Having established that the items in $L_{j-1} \setminus L_j$ are no longer
exceptions, $V$ also wants to ensure that the frequencies of these items were
reported correctly in $L_0$. 
To do so, $V$ %derive a vector $f^{(j)}$ of the claimed
%frequencies in $L_{j-1} \setminus L_j$, and 
run the 
\subftwo\ \ascheme\ over the vector $f - f^*$ as mapped by
$h_j$ to $r$ buckets, based on the $z^{(j)}$ indicator vector. 
The result is zero if and only if $f_i = f^{(j)}_i$ for all $i$
where $z^{(j)}_i=1$.

% repeat these in the annotation when they are
%needed, and $V$ can process these in a streaming fashion. 
%In parallel, $V$ can maintain a constant number of fingerprints to
%ensure that the repetitions of $L_j$ are in full agreement. 

The stages continue until $L_j = \emptyset$, and there are no more
exceptions. 
Provided all \aschemes\ conclude correctly, and the number of stages to
reach $L_j = \emptyset$ is at most $t$, $V$ can accept the result, and
output 1 for the answer to the \multiindex\ decision problem. 

Lastly, note that $V$ does not need to explicitly store any of the
lists $L_j$. In fact, $P$ can implicitly specify all of the lists $L_j$ while playing the list $L_0$:
for each item $i \in L_0$, he provides a number $j$, thereby implicitly claiming that 
$i \in L_{j'}$ for $j' \leq j$, and $i \not\in L_{j'}$ for $j' > j$.

%Let $i^*$ denote the final stage. The accuracy of the claimed frequencies of items in $L_{i^*-1}$ follows 
%from the soundness of the \subinj\ streams in the strict turnstile model, as well as the soundness of the \subftwo\ scheme.
%Notice we are exploiting the fact that $L_{i^*} = \emptyset$ in order to conclude that the 
%input to the \subinj\ scheme is applied to an input in the strict turnstile model.

%Now assume by way of induction that the frequencies of all items in list $L_{j}$
%are as claimed. This implies that the input to the \subinj\ scheme invoked at stage $j$ is also in the strict turnstile model.
%The accuracy of the claimed frequencies of items in $L_{i^*-1}$ then follows 
%from the soundness of the \subinj\ streams in the strict turnstile model, as well as the soundness of the \subftwo\ scheme.
%This completes the induction and the proof of soundness.

\medskip 
\noindent \textbf{Analysis of costs.}
If $h_j$ is chosen at random from a pairwise independent hash family,
the probability an item $i$ in $L_{j-1}$ is involved in a collision with the original stream $f$ under $h_j$
is  $O(\s/r)=O(c_v^{-1/2})$. 
Consider the probability that any item $i$ survives as an exception to
stage $t$. 
The probability of this is  $O(c_v^{-t/2})$, and summed over all $\ell$ items,
the expected number is 
$O(\ell c_v^{-t/2})$. 
Invoking Markov's inequality, with constant probability it suffices to
set $t = O(\log_{c_v} \ell )$ to ensure that we need at most $t$ stages  before no more exceptions need to be reported.

In stage $j$, the \subinj\ and \subftwo\ schemes cost
$(\s c_v^{-1/2} \log n,c_v \log n)$. 
 Summing over the $t$ stages, we achieve for any $c_v > 1$ an 
$(\s c_v^{-1/2} \log(n) \cdot \log_{c_v}(\s),c_v \log(n) \cdot \log_{c_v}(\s))$-\ascheme\ as claimed in the statement of Lemma \ref{lemma:multiindex}.
 
 \medskip
 \noindent \textbf{Formal Proof of Soundness.}
The soundness error of the protocol can be bounded by the probability any invocation of the \subinj\ \ascheme\ or the \subftwo\ \ascheme\ 
returns an incorrect answer. The soundness errors of both the \subinj\ \ascheme\ and the \subftwo\ \ascheme\ can be made smaller than $\frac{1}{r^c}$ for any constant $c>0$,
and therefore a union bound over all $t = O(\log_{c_v} \ell )$ invocations of each protocol implies that with high probability, no invocation of either \ascheme\ returns an incorrect answer. 

\subsection{Implications of the Online Scheme for Frequency Moments}

Our online scheme for \fk\ in Theorem \ref{thm:of2} has a number of important consequences.

\medskip \noindent \textbf{Inner Product and Hamming Distance.} Chakrabarti \etal\ \cite{icalp} 
point out that computing inner products and Hamming
Distance can be directly reduced to (exact) computation of the second
Frequency Moment $F_2$,
and so Theorems \ref{thm:pf2} and \ref{thm:of2} immediately yield \aschemes\ for these problems of identical cost. 

\medskip \noindent \textbf{An improved scheme for $\phi$-\hh.} 
We can use Lemma \ref{lemma:multiindex} to yield an online scheme for the $\phi$-\hh\ problem. 

\begin{corollary}
\label{cor:hh2}
For all $\anncost,c_v$ such that $\anncost c_v \geq \s \log n$, there is an \oscheme{(\anncost \log n \cdot \log_{c_v}(\s)+ \phi^{-1} \log n,$
$c_v \log n \log_{c_v}(\s))} for solving $\phi$-\hh\ in the strict turnstile update model. %Any online or prescient $(h, v)$-protocol for computing the $\phi$-heavy hitters for $\phi=1/m$ requires $\anncost c_v = \Omega(m)$.
\end{corollary}

%\begin{proof} 
Corollary \ref{cor:hh2} follows from a similar analysis to Corollary \ref{cor:hh}. \cite[Theorem 6.1]{icalp} describes how to reduce $\phi$-\hh\ to demonstrating
the frequencies of $O(\phi^{-1})$ items in a derived stream. Moreover, the derived stream has sparsity $O(\s \log n)$
if the original stream has sparsity $\s$. 
We use the \multiindex\ \ascheme\ of Lemma \ref{lemma:multiindex} to verify these claimed frequencies.

\medskip \noindent \textbf{Frequency-based functions.}
Chakrabarti \etal\ \cite[Theorem 4.5]{icalp} also explain how to extend the sum-check \ascheme\ of Proposition \ref{prop:dense} to efficiently
compute arbitrary \emph{frequency-based functions}, which are functions of the form $F(\mathbf{x}) = \sum_{i \in [n]} g(f_i(\bx))$ for an arbitrary $g: (-[\length] \cup [\length]) \rightarrow \mathbb{Z}$. 
A similar but more involved extension applies in our setting, by replacing the dense \fk\ \ascheme\ implied by Proposition \ref{prop:dense} with the dense frequency-based functions \ascheme\ of \cite[Theorem 4.5]{icalp}.
%For brevity, we only spell out the consequences for a particularly well-studied frequency-based function, the distinct elements problem.
We spell out the details below, restricting ourselves to the prescient case for brevity; an online \ascheme\ with essentially identical costs follows by using 
the ideas underlying Theorem \ref{thm:of2}.

\begin{corollary}
Let  $\fn(\mathbf{x}) = \sum_{i \in [n]} g(f_i(\bx))$ be a frequency-based function. Then there is a prescient $(\length^{3/4}\log n,$
$\length^{3/4}\log n)$-scheme
for computing $\fn(\mathbf{x})$ in the strict unit-update turnstile model. This scheme satisfies perfect completeness.
\end{corollary}

\begin{proof}
We use a natural modification of the frequency-based functions
\ascheme\ of \cite[Theorem 4.5]{icalp}. 
$P$ specifies a hash function $h$ at the start of the stream mapping the
universe $[n]$ into $[\length^{5/4}]$; $P$ chooses $h$ 
to be injective on the set of items that have non-zero frequency at the end of the stream. Using the perfect hash functions of Fredman and Koml\'{o}s \cite{FredmanK84}, 
$h$ can be represented with $O(\length^2/r\log n) = O(\length^{3/4}\log n)$ bits.  
$V$ stores $h$ explicitly. After the stream is observed, $P$ and $V$ run the $\phi$-\hh\ \ascheme\ of Corollary~\ref{cor:hh2},
with $\phi=\length^{-1/4}$. Using the fact that $\sum_i f_i < \length$, by setting the parameters of Corollary~\ref{cor:hh2} appropriately 
we can ensure that this part of the \ascheme\ requires annotation length $O(\length^{3/4} \log n)$ and has space cost $O(\length^{3/4} \log n)$.
This scheme also allows $V$ to determine the exact frequencies of the items in $H$,
allowing $V$ to compute $\text{cont}(H) := \sum_{i \in H} g(f_i(\bx))$, which
gives the contribution of the items in $H$ to the output $F(\mathbf{x})$. 
Moreover, whenever $V$ learns the frequency $f_i$ of an item in $i \in H$, $V$ treats this as a deletion of $f_i$ occurrences of item $i$, thereby obtaining a derived stream $\mathbf{z}$ in which all frequencies have absolute value at most $\length^{1/4}$. 

$P$ and $V$ now run the \emph{polynomial-agreement} \ascheme\ that was first presented in \cite[Theorem 4.6]{esa} on the ``mapped-down'' input $h(\mathbf{z})$ over the universe $[\length^{5/4}]$. 
For any $\anncost c_v \geq r$, the polynomial agreement \ascheme\ can achieve cost $(F_{\text{max}}(\mathbf{z}) \anncost \log n, c_v \log n)$, where $F_{\text{max}}(\mathbf{z})$ denotes $\max_i |f_i(\mathbf{z})|$,
the largest frequency in absolute value of any item.
 Setting $c_v=\length^{3/4}$ and $\anncost =\length^{1/4}$, we obtain
a \pscheme{(\length^{3/4} \log n, \length^{3/4} \log n)} as claimed.
$V$ computes the final answer as $\fn(\bx) = \text{cont}(H) + \fn(h(\mathbf{z})) - |H|g(0)$.

The final issue is that $V$ needs to verify that $h$ is actually
injective over the items that appear in $\mathbf{x}$. 
$V$ can accomplish this using
the \inject\ \ascheme\ of Lemma \ref{lemma:injection}. 
This does not affect the asymptotic costs of our \ascheme, as the
\inject\ \ascheme\ can support annotation cost $\anncost \log r$ and space
cost $c_v \log r$ for any $\anncost c_v = \Omega(\length^{5/4})$.  
\end{proof}

Finally, we provide one additional corollary, which describes a protocol that will be useful in the next section when building graph \aschemes.  

\begin{theorem} \label{thm:subset}
 Let $X,Y \subseteq [n]$ be sets with $|X| \leq |Y| \leq \s$. Then given a stream in the strict turnstile update model with elements of $X$ and $Y$ arbitrarily interleaved, 
 there is an \oscheme{(\s c_v^{-1/2} \cdot \log(n) \cdot  \log_{c_v}(\s), c_v \cdot \log(n) \cdot \log_{c_v}(\s))} for determining whether $X \subseteq Y$ for any $c_v>1$.
 \end{theorem}

 \begin{proof} If $X \not\subseteq Y$, $P$ can specify an $x \in X \setminus Y$ and prove that $x$ is indeed in $X$ and not $Y$ with two point queries using the \ascheme\
 of Theorem \ref{thm:pointquery}. For the other case, 
 Chakrabarti \etal\ show how to directly reduce the case $X \subseteq Y$ to computation of frequency moments~\cite{icalp}. The claimed costs follow from Theorem \ref{thm:of2}.
 \end{proof}
   
 \begin{table}[tb]
\centering
\begin{tabular}{r@{}l|c|c|c}
\multicolumn{2}{c|}{\aScheme\ Costs} &  Completeness & Online/Prescient & Source \\ \hline
$(|X|\log n,\, \log n)$ && Perfect &  Prescient &  \cite{icalp}\\
 $(\anncost \log n,\, c_v \log n)$&: $\anncost c_v \geq n$ & Perfect & Online & \cite{icalp}\\
%\hline
$(\s \log n,\, \log n)$ && Perfect & Online &  \cite{icalp}\\
%\hline
%\hline
 $( \s c_v^{-1/2}\log_{c_v}(\s) \log n,\, c_v \log n \log_{c_v} \s )$ &: $c_v >1$	&Imperfect & Online & Theorem \ref{thm:subset}\\
%\hline
\end{tabular}
\caption{Comparison of our {\sc Subset} \ascheme\ to prior work. Ours is the first online \ascheme\ to achieve annotation length and space usage 
that are both sublinear in $\s$ for $\s \ll \sqrt{n}$, and strictly improves over the online MA communication cost of prior protocols whenever $\s=o(n)$.}
\label{tabsubset}
%\vspace{-0.15in}
\end{table}

Table~\ref{tabsubset} provides a comparison of schemes for the {\sc Subset}
problem in the dense and sparse cases.

\section{Graph Problems} \label{sec:graphs}

We now describe some applications of the techniques developed above to graph problems.
The main purpose of this section is to demonstrate that the techniques developed within the \fk\ and \kdisj\
\aschemes\ are broadly applicable to a range of settings.

We begin with several non-trivial graph \aschemes\ that are direct consequences of the Subset \ascheme\ of Theorem \ref{thm:subset}.
Recall that our definition of a scheme for a function $\fn$ requires a convincing proof of the value of $\fn(\bx)$
\emph{for all values $\fn(\bx)$}. This is stricter than the traditional definition of interactive proofs for decision problems,
which just require that if $\fn(\bx)=1$ then there is some prover that will cause the verifier to accept with high probability,
and if $\fn(\bx)=0$ there is no such prover. Here, we consider a relaxed definition
of schemes that is in the spirit of the traditional definition. We require only that a scheme $\mathcal{A}=(\help, V)$ satisfy:

\begin{enumerate}
  \item For all $\bx$ s.t. $\fn(\bx)=1$, we have $\Pr_{r_P, r_V}[\out{V}{\bx^{\rhelp},r_V} \neq 1] \le 1/3$.
  \item For all $\bx$ s.t. $\fn(\bx)=0$, $\help'=(\help_1', \help_2',
  \ldots, \help_\length') \in (\b^*)^\length $, we have
    $\Pr_{r_V}[\out{V}{\bx^{\help'},r_V} = 1] \le 1/3$.
\end{enumerate}

\begin{theorem} \label{thm:relaxed}
Under the above relaxed definition of a scheme, each of the problems \textsc{perfect-matching},
\textsc{connectivity}, and \textsc{non-bipartiteness} 
has an $(n \log n + \s c_v^{-1/2} \log n\, \log_{c_v}\s,~ c_v \log n\, \log_{c_v}\s)$-scheme
on graphs with $n$ vertices and $\s$ edges for all $c_v >1$. 
All three \aschemes\ work in the strict turnstile update model and improve over prior work if $c_v = \omega(\log^2 \s)$ and $c_v=o(\s)$. 
\end{theorem}

\begin{proof}
In the case of perfect matching, the prover can prove a perfect matching exists by 
sending a matching $\mathcal{M}$, which requires $n \log n$ bits of annotation. In order to prove $\mathcal{M}$ is a valid perfect matching,
$P$ needs to prove that every node appears in exactly one edge of $\mathcal{M}$, and that $\mathcal{M} \subseteq E$, where $E$ is the set
of edges appearing in the stream. $V$ can check the first condition by comparing a fingerprint of the nodes in $\mathcal{M}$
to a fingerprint of the set $\{1, \dots, n\}$. $V$ can check that $\mathcal{M} \subseteq E$ using Theorem \ref{thm:subset}.

In the case of connectivity, the prover demonstrates the graph is connected by specifying a spanning tree $T$. $V$ needs to check $T$ is spanning, which can be done
as in \cite[Theorem 7.7]{icalp}, and needs to check that $T \subseteq E$, which can be done using Theorem \ref{thm:subset}.

In the case of non-bipartiteness, $P$ demonstrates an odd cycle $C$. $V$ needs to check $C$ is a cycle, $C$ has an odd number of edges,
and that $C \subseteq E$. The first condition can be checked by requiring $P$ to play the edges of $C$ in the natural order. The second condition
can be checked by counting. The third condition can be checked using Theorem \ref{thm:subset}.
\end{proof}

\medskip
\noindent \textbf{Counting Triangles.}
Returning to our strict definition of a scheme, we give an online \ascheme\ for counting the number of triangles in a graph. 

\begin{table}[tb]
\centering
\begin{tabular}{r@{}l|c|c|c|}
\multicolumn{2}{c|}{\aScheme\ Costs} & Completeness & Online/Prescient & Source \\ \hline
 $(\anncost \log n,c_v \log n)$&: $\anncost c_v \geq n^3$ & Perfect & Online & \cite{icalp}\\
%\hline
$(n^2\log n, \log n)$ && Perfect & Online &  \cite{icalp}\\
%\hline
$(\anncost \log^2 n,c_v \log^2 n)$&: $\anncost=\s n c_v^{-1/2}$	& Imperfect & Online & Theorem \ref{thm:triangle}\\
%\hline
\end{tabular}
\caption{Comparison of prior work to our \ascheme\ for counting the number of triangles in a graph with $n$ nodes and $\s$ edges.
For concreteness, notice that by setting $c_v=n$, Theorem \ref{thm:triangle} achieves a $(\s n^{1/2} \log^2 n, n \log^2 n)$-\ascheme, which improves over prior work
as long as $\s \ll n^{3/2}$.}
\label{tabtriangle}
%\vspace{-0.15in}
\end{table}

\begin{theorem} \label{thm:triangle} For any $c_v > 1$, there is an \oscheme{(\anncost \log n \log \s,\, c_v \log n \log \s)}, with imperfect completeness, for counting
the number of triangles in a graph on $n$ nodes and $\s$ edges, where
$\anncost=\s n c_v^{-1/2}$. 
The \ascheme\ is valid in the strict turnstile update model.
\end{theorem}
\begin{proof}
Chakrabarti \etal\ \cite[Theorem 7.4]{icalp} show how to reduce counting the number of triangles in a graph to computing the first three frequency moments of a derived stream.
The derived stream has sparsity $\s(n-2)$. Using the online scheme of Theorem \ref{thm:of2} to compute the relevant frequency moments of the derived stream yields the claimed bounds.
\end{proof}

 The \ascheme\ of Theorem \ref{thm:triangle} should be compared to the \scheme{(n^2, \log n)} from \cite[Theorem 7.2]{icalp} based on matrix multiplication, referenced in Row 2 of Table \ref{tabtriangle}
 and the \scheme{(h, v)} for any $\anncost c_v \geq n^3$ from \cite[Theorem
   7.3]{icalp}, referenced in Row 1 of Table \ref{tabtriangle}. 
To compare 
 to the former, notice that Theorem \ref{thm:triangle} yields a \scheme{(\anncost \log^2 n,c_v \log^2 n)} with $ \anncost < n^2$ as long as $\s < n\sqrt{c_v}$. 
 To compare to the latter, note that 
 in our new \ascheme,
 $\anncost c_v = \s n c_v^{1/2}$, which is less than $n^3$ as long as $c_v^{1/2} < \frac{n^2}{\s}$.
In particular, if we set $c_v=n$, then Theorem \ref{thm:triangle}
improves over both old \aschemes\ as long as $\s < n^{3/2}$.

Unfortunately, Theorem \ref{thm:triangle} does not yield a non-trivial
MA-protocol for showing no triangle exists.  
%though for a rather limited range of parameters.
Indeed, equalizing annotation length and space usage in our new protocol occurs by
setting both quantities to $(\s n)^{2/3}$. 
But $\Omega\left(( \s n)^{2/3}\right) < \s$ only when $\s > n^2$, which is to say that the MA communication complexity of this protocol
is always larger than $\s$, a cost that can be achieved by the trivial MA protocol where Merlin is ignored and Alice
just sends her whole input to Bob. 
That is, the interest in the new protocol is that it can lower the
space usage of $V$ to less than $\s$ without drastically blowing up the message length of $P$ to $n^2$ as in the matrix-multiplication
based protocol from \cite{icalp}.

%This improves on 
%the protocol from Annotations when 
%$$(mn)^{2/3} < n^{3/2} \Longleftrightarrow m^{2/3} < n^{3/2 - 2/3} = n^{5/6} \Longleftrightarrow 
%m < (n^{5/6})^{3/2} = n^{5/4}.$$ But in this setting of parameters, $h, v >$ are larger than $m$
%So in summary, we improve over the MA communication protocol in the Annotations paper
%when $m < n^{5/4}$.

\section{Non-strict Turnstile Update Model} \label{sec:genturnstile}

All schemes in Sections \ref{sec:disj} and \ref{sec:of2} work in the strict turnstile update model. 
The reason these schemes require this update model is that they use
the \inject\ and \subinj\ schemes of Lemmata \ref{lemma:injection} and \ref{lemma:strict} as sub-routines, and these sub-routines assume the strict turnstile update model. 

In this section, we consider two ways to circumvent this issue. 
To focus the discussion, we concentrate on the online \fk\ protocol of Theorem \ref{thm:of2}. 
%However, we are able to apply our \ascheme\ to input streams which
%contain negative values, since the input to \inject\ does {\em not}
%contain negative values. 
%To clarify, we do not know how to modify to handle
%deletions in the stream of updates $(x_i, b_i)$. Fortunately, we do
%not run the \inject\ protocol on the raw stream who's frequency moment
%we want to compute, but only on a derived stream, so there is no issue
%when dealing with deletions in the raw stream. 

\subsection{An Online Scheme}
One simple method for handling streams in the non-strict turnstile update model is the following. 
We use the scheme of Theorem \ref{thm:of2}, but within the \subinj\ sub-routine, we treat deletions of items in
the input stream as {\em insertions} of items into the derived stream of
$(x_i, b_i, \delta_i)$ updates. 
This ensures that the \inject\ and \subinj\ schemes 
correctly output 1 if the derived stream is a subinjection (and the
remainder of the scheme computes the correct answer on the original stream). 
However it increases the expected number of collisions
under the universe-reduction mappings $h_i$, from $\s\cdot
|L_{i-1}|/r$ to $\footprint \cdot |L_{i-1}|/r$.
The result is that we achieve the same costs as Theorem \ref{thm:of2}, except the costs depend on to the
stream footprint $\footprint$ rather than the stream sparsity
$\s$ (see Section~\ref{sec:notation}).

\begin{corollary}  \label{cor:footprint}
For any $c_v > 1$, there is a $(k^2 \footprint c_v^{-1/2} \cdot
\log(n) \cdot \log_{c_v}(\footprint), k c_v \cdot \log(n) \cdot
\log_{c_v}(\footprint))$ online \ascheme\ for \fk\ in the non-strict turnstile update model
over a stream with footprint $\footprint$ over a universe of size $n$.
%For any $1 < v < M$, there is an \oscheme{(k^2h\log n, kv\log n)}
%for computing \fk\ over
%a stream of length $m$ over a universe of size $n$ in the non-strict turnstile update model. This \ascheme\ satisfies imperfect completeness.  %This is not quite as good as the$k$-Disjointness protocol, where we have an $(h, v)$ protocol for $\anncost c_v = m\sqrt{v}$ which is always less than $m^{3/2}$ if $v < m$.}
\end{corollary}

\subsection{An Online AMA \aScheme}
\label{sec:amastream}
In this section, we describe an AMA scheme for the \inject\
problem that works in the non-strict turnstile stream
update model i.e., the input may define a frequency vector where some
elements end with negative frequency. 
The \ascheme\ for \inject\ of Lemma \ref{lemma:injection} breaks down here, since there may be
some cases where the checks performed by the protocol indicate that
a bucket is pure, when this is not the case: cancellations of item
weights in the bucket may give the appearance of purity. 
To address this, we use public randomness, thereby yielding an AMA scheme. 
In essence, the verifier asks the prover to demonstrate the purity of
each  of the $r$ buckets via fingerprints of the bucket contents. 
However, if we allow the prover to choose the fingerprint function, $P$
could pick a function which leads to false conclusions. 
Instead, $V$ chooses the fingerprinting function using public randomness. 
The players then execute a new \inject\ protocol using
the data remapped under the fingerprint function, which is intended to
convince $V$ of the purity of the buckets. 
This then allows us to construct protocols with costs that depend on the stream sparsity $\s$
rather than the footprint $\footprint$ as in Corollary \ref{cor:footprint}.

In detail, the new AMA \ascheme\  proceeds as follows. 
Consider the \inject\ problem as defined in
Definition~\ref{def:inject}, but generalized to allow items with
arbitrary integer counts. 
Consider again a bucket $b$, and for $1 \leq j \leq \log n$
define $b^{j=\ell}$ to be the frequency vector of the subset of stream updates $(x_k, b, \delta_k)$ placing items into bucket $b$, subject
to the restriction that the $j$'th bit of $x_k$ is equal to $\ell$.
We observe the following property: if bucket 
$b$ is pure, then one of $b^{j=0}$ and $b^{j=1}$ must be the zero
vector $\mathbf{0}$, for each $j$.
Moreover, if $b$ is not pure, then there exists a $j$ such that both 
$b^{j=0}$ and $b^{j=1}$ are not  the zero vector. 

A natural way to compactly test whether these vectors are equal to
zero (probabilistically) is to use fingerprinting (discussed in Section~\ref{sec:notation}).
The verifier $V$ could do this unaided for a single bucket, but we wish
to run this test in parallel for $r$ buckets. At a high level, we achieve this as follows.
Given a stream of updates $(x_k, b, \delta_k)$, we define two vectors $z$ and $o$ of length $r\log n$, 
such that each coordinate of $z$ and $o$ corresponds to a (bucket, coordinate) pair $(b, j) \in [r] \times [\log n]$.
In more detail, we will define $z$ and $o$ such that for each bucket $b$ and coordinate $j \in [\log n]$, 
the $(b, j)$th entry of $z$ is a
fingerprint of the vector $b^{j=0}$, and the $(b, j)$th entry of $o$ is a fingerprint of the vector $b^{j=1}$. 

We choose the fingerprinting functions to satisfy two properties.
\begin{enumerate}
\item The fingerprint of the all-zeros vector $\mathbf{0}$ is always 0. This ensures
that if all buckets are pure, then the inner product of $z$ and $o$ is 0, as $z_{b, j} \cdot o_{b, j}$ is 0 for all pairs $(b, j) \in [r] \times [\log n]$.

\item If there is an impure bucket, then the inner product of $z$ and
  $o$ will be non-zero with high probability over the choice of
  fingerprint functions.  
%$\alpha$ and $\beta$.
\end{enumerate}

Therefore, in order to determine whether the stream defines an injection, it suffices to compute $\sum_{(b, j) \in [r] \times [\log n]} z_{b, j} \cdot o_{b, j}$, 
which can be computed using Proposition \ref{prop:dense} with annotation length $\anncost \log n$ and space cost $c_v \log n$ for any $\anncost \cdot c_v \geq r\log n$. 

The idea allowing us to achieve the second property is as follows. If bucket $b$ is impure, then there is at least one coordinate $j \in [\log n]$ such that $b^{j=0}$ and $b^{j=1}$ are both not equal to the all-zeros vector $\mathbf{0}$.
By basic properties of fingerprints,  this ensures that both $z_{b,
  j}$ and $o_{b,j}$ are non-zero with high probability over the choice
of fingerprint functions. 
%$\alpha$ and $\beta$. 
Moreover, we choose the fingerprinting functions
in such a way that non-zero terms in the sum $\sum_{(b, j) \in [r] \times [\log n]} z_{b, j} \cdot o_{b,j}$ are unlikely to ``cancel out'' to zero.

Consequently, we can state an analog of Lemma~\ref{lemma:injection}.

\begin{lemma}
\label{lemma:geninjection}
For any $\anncost c_v \geq r \log n$, there is an \oscheme{(\anncost \log n,
  c_v\log n)} for determining whether a stream in the non-strict turnstile
model is an injection.
\end{lemma}

\begin{proof}
%\medskip \noindent \textbf{Formal Details.}
Let $\mathbb{F}_q$ be a finite field of size $q=\poly(n)$, where 
the subsequent analysis determines the required magnitude of $q$. 
%we specify the precise value $q$ later.
$V$ uses public randomness to choose two field elements $\alpha$, and $\beta$ uniformly at random from $\mathbb{F}_q$.
For each bucket $b \in [r]$, and each coordinate $j \in [\log n]$, we define two ``fingerprinting'' functions $g_{b, j, \alpha}$ and $g_{b, j, \beta}$ mapping an $n$-dimensional frequency vector $\mathbb{F}$ as follows:
$$g_{b, j, \alpha}(\bx) = \alpha^{n(b \cdot \log n + j)} \sum_{\ell \in [n]} \bx_{\ell} \alpha^{\ell},$$ and $$g_{b, j, \beta}(\bx) = \beta^{n(b \cdot \log n + j)} \sum_{\ell \in [n]} \bx_{\ell} \beta^{\ell},$$

\noindent
where each entry $\mathbf{x}_\ell$ of $\bx$ is treated as an element of $\mathbb{F}$ in the natural manner.

We now (conceptually) construct two vectors $z$ and $o$ of dimension $r \log n$,
where for each $(b, j) \in [r] \times [\log n]$,
$z_{b, j} = g_{b, j, \alpha}(b^{j=0})$ and 
$o_{b, j} = g_{b, j}(b_i^{j=1})$. 
That is, the $(b, j)$th entry of $z$ equals the fingerprint of the
frequency vector of items mapping to bucket $b$ with a 0 in the $j$th
bit of their binary representation.
Observe that $g_{b, j, \alpha}(\mathbf{0})=g_{b, j, \beta}(\mathbf{0})=0$ for all $(b, j) \in [r] \times [\log n]$, as required by Property 1 above.

We now show that Property 2 holds, i.e.  if there is an impure bucket, then the inner product of $z$ and $o$ will be non-zero with high probability over the choice of $\alpha$ and $\beta$. In
the following, for an item $\ell \in [n]$ and bucket $b \in [r]$, we let $f_{\ell}(b)$ denote the frequency with which item $\ell$ is mapped to bucket $b$, and we let $\ell_j$ denote the $j$'th bit in the binary representation of $\ell$.
We can write the inner product of $z$ and $o$ as

\begin{align*} & \sum_{(b, j) \in [r] \times [\log n]} g_{b, j, \alpha}(b^{j=0})  g_{b, j, \beta}(b^{j=1})\\
&  =  \sum_{(b, j) \in [r] \times [\log n]} \alpha^{n(b \cdot \log n + j)} \beta^{n(b \cdot \log n + j)}\left(\sum_{\ell \in [n], \ell_j = 0} f_{\ell}(b) \alpha^{\ell} \right) \left(\sum_{\ell \in [n], \ell_j = 1} f_{\ell}(b) \beta^{\ell} \right)\\
 & =  \sum_{(b, j) \in [r] \times [\log n]} \alpha^{n(b \cdot \log n + j)} \beta^{n(b \cdot \log n + j)} \sum_{ (\ell, \ell'): \ell_j=0, \ell'_j=1} f_{\ell}(b) f_{\ell'}(b) \alpha^{\ell} \beta^{\ell'}\end{align*}

We therefore see that the inner product of $z$ and $o$ is a polynomial in $\alpha$ and $\beta$ of total degree $n^2r\log n$ in each variable. 
Moreover, the coefficient of the term $\alpha^{n(b \cdot \log n + j) + \ell} \beta^{n(b \cdot \log n + j) + \ell'}$ is precisely $f_{\ell}(b) \cdot f_{\ell'}(b)$ if $\ell_j=0$ and $\ell'_j=1$, and is 0 otherwise.

Recall that if bucket $b$ is not pure, then there is at least one coordinate $j \in [\log n]$, and items $\ell, \ell' \in [n]$ with $\ell_j=0$ and $\ell'_j=1$, such that $f_{\ell}(b)\neq 0$ and $f_{\ell'}(b) \neq 0$.
The above analysis implies that $z \cdot o$ is a \emph{non-zero} polynomial in $\alpha$ and $\beta$, as the coefficient of $\alpha^{n(b \cdot \log n + j) + \ell}\beta^{n(b \cdot \log n + j) + \ell'}$ is non-zero.
Hence, by the Schwartz-Zippel lemma, the probability over a random choice of $\alpha$ and $\beta$ that $z \cdot o = 0$ is at most $n^2r\log n/q$.
Setting $q$ to be polynomial in $n$, there is only negligible probability (over the choice of $\alpha$ and $\beta$) that $z \cdot o$ is zero if the stream is not an injection.

Finally, notice that the verifier can apply the scheme of Proposition \ref{prop:dense} to compute $\sum_{(b, j) \in [r] \times [\log n]} z_{b, j} \cdot o_{b, j}$, as each stream update
$(x_k, b, \delta_k)$ can be treated as $\log n$ updates to the vectors $z$ and $o$. For example,
if the $j$th bit of $x_k$ is 0, then update $(x_k, b, \delta_k)$ causes $z_{b, j}$ to be incremented by $\delta_k \cdot \alpha^{n(b \cdot \log n + j)+x_k}$.
\end{proof}

%degree $2n$ polynomial, where 
%the function $g$ is a polynomial of bounded degree $d$ in a 
%value $\alpha$. \textbf{Still need to correct this analysis... It might require introducing the specifics of the fingerprint function (since the proof really exploits the Schwartz-Zippel lemma i.e. leverages the algebratic
%properties of the fingerprint function we are really using)}
%Hence, the overall result can also be seen as a polynomial of degree
%at most $2d$ in $\alpha$, and hence provided the field is large
%enough, the probability of this error is negligible. 

\medskip
\noindent \textbf{Applications.} We can apply this online \ascheme\ to compute Frequency Moments (and Inner
Product, Hamming Distance, Heavy Hitters etc.) over sparse data in the
non-strict turnstile update model.  
%The important detail is that prover does {\em not} get to specify the
%input to the problem after he knows the parameters of the fingerprint
%functions. 
%Rather, the input is defined by the input stream, over which $P$ has
%no control. 
%If $P$ did control the input, then he could use his knowledge of the
%fingerprint functions to engineer cases where the input yielded
%non-zero vectors whose fingerprint was zero, and thus violate the
%requirements. 
%For the above \aschemes described in Section~\ref{sec:of2}, 
%the input is given by the stream. 
%$P$ is allowed to indicate which buckets suffer collisions under his
%hash function, but this is not enough to cheat. 
The costs of the resulting online AMA scheme are similar
to the costs of the online schemes for the same problems developed in previous sections.
%identical to those of the bounds from previous sections for the strict turnstile update model. 
%Notice that while the \emph{sparsity} of the vectors $z$ and $o$ defined within the modified \inject\ protocol is a $\log n$ factor
%larger than the sparsity of the original stream (since each update $(x_k, b, \delta_k)$ is treated as an update to $\log n$ coordinates of $z$ and $o$),
%the costs of the \inject\ protocol depend only on the number of buckets $r$, and hence are unaffected by this $\log n$ factor increase in sparsity.
The only difference is that we have scaled $\s$ up by a $\log n$
factor, to account for the fact that within the new AMA sub-scheme for \inject, we must run the dense protocol
of Proposition \ref{prop:dense} on vectors $z$ and $o$ of length $r
\log n$, rather than on vectors of length $r$ as in prior sections,
and substitute the bounds from Lemma~\ref{lemma:geninjection}. 
For example, the analog of Theorem~\ref{thm:of2} is that 
for any $c_v > 1$, there is a $(k^2 \s c_v^{-1/2} \cdot \log^2(n)
\cdot \log_{c_v}(\s), kc_v \cdot \log(n) \cdot \log_{c_v}(\s))$ online
AMA \ascheme\ for \fk\ in the non-strict turnstile model.
%for a stream of sparsity $\s$ over a universe of size $n$. 
%Any online $(\anncost, c_v)$ scheme for this problem with $\anncost \geq \log n$ requires $\anncost \cdot c_v = \Omega(\s \log(n/\s))$.
%\textbf{Also insert a section on unbalanced disjointness?}

\section{Conclusion} \label{sec:conclusion}

We have presented a number of protocols in the annotated data streaming model that
for the first time allows both the annotation length and the space usage of the verifier to be sublinear in the stream sparsity,
rather than just the size of the data universe. Our protocols substantially improve on the applicability of prior work in natural settings where 
data streams are defined over very large universes, such as IP packet flows and sparse graph data. 

A number of interesting questions remain for future work. The biggest open question is to determine the precise dependence on the stream sparsity
in problems such as \kdisj\ and frequency moments. When setting the annotation length and the space usage of the verifier to be equal,
our protocols have cost roughly $\s^{2/3}$, where $\s$ is the sparsity of the data stream. The best known lower bound is roughly $\s^{1/2}$. 
We conjecture that our upper bound is tight up to logarithmic factors, but proving any Merlin-Arthur communication lower bound larger than 
$\s^{1/2}$ will require new lower bound techniques in communication complexity. 
Another interesting open question is to give improved protocols for multiplying an $n \times n$ matrix $A$ by a vector $x$, when $A$ is sparse
(i.e., has $o(n^2)$ non-zero entries), 
but $x$ may be dense.
Achieving this would yield improved protocols for
proving disconnectedness, bipartiteness, or the non-existence of a perfect matching in a bipartite graph.
Currently we do not know of any protocols for these problems that leverage graph sparsity in any way.

%\begin{enumerate}
%\item
%As already mentioned: come up with protocols for
%proving disconnectedness, bipartiteness, or the non-existence of a perfect matching in a bipartite graph,
%for which both space usage and annotation size is sublinear in the stream length $\length$.
%It would be sufficient to give such a protocol for matrix-vector multiplication when the matrix $A$ is sparse ($m = o(n^2)$ entries),
%but the vector $x$ is not ($x$ may have $\Omega(n)$ non-zero entries).

%\item There is a gap between the upper and lower bounds for $k$-\disj\ and $F_k$ (prescient or online). Close this gap.

%\item Close the remaining $\log n$ gap between the upper and lower bounds for the online MA communication complexity
%of \idx. (sparse or dense)?
%\end{enumerate}

% --------> BIBLIO <--------

% --------> APPENDICES <--------
\appendix

\section{An Online AMA Lower Bound for $(m, 2^{\sqrt{m}})$-Sparse \idx} \label{app:AMA}

We prove that the online $\tilde{O}(\sqrt{m})$ protocol for the $(m, 2^{\sqrt{m}})$-Sparse \idx\ problem is essentially optimal. 
Our lower bound follows from a natural variant of the reduction in Theorem \ref{thm:lb}. That is, we turn an online AMA protocol for the  $(m, 2^{\sqrt{m}})$-Sparse \idx\ Problem into an online 
MAMA protocol for the dense \idx\ Problem.
We then invoke a lower bound on the online MAMA communication complexity of \idx\ Problem due to Klauck and Prakash \cite{klauck}.\footnote{
Like the lower bound of Lemma \ref{lemma:denseindex}, the lower bound of Klauck and Prakash was originally proved in the communication 
model in which Merlin cannot send any message to Alice. However, the proof
easily extends to our online MA communication model (where Merlin can send a message to Alice, but that message cannot depend on Bob's input).} %We describe the details in Appendix \ref{app:online}.

\begin{theorem} The online AMA protocol complexity of the $(m, 2^{\sqrt{m}})$-Sparse \idx\ problem is $\tilde{\Omega}(\sqrt{m})$. \end{theorem}
\begin{proof}
Let $n = 2^{\sqrt{m}}$.
Assume we have an online AMA communication
protocol $\cP$ for $(m, n)$-sparse \idx\ with  $\hcost(\cP)=\Omega(\sqrt{m})$.
We describe how to use this protocol for the sparse \idx\
problem to design one for the dense \idx\ problem on vectors of length $n'=m \log(n/m) = \Omega\left(m^{3/2}\right)$.

Let $k=\log(n/m)$. As in the proof of Theorem \ref{thm:lb}, given an input $x$ to the dense \idx\ problem, Alice partitions $x$ into $n'/k$ blocks of length $k$, and constructs a vector $y$
of Hamming weight $n'/k$ over a universe of size $(n'/k) \cdot 2^k$ as follows.
She replaces each block $B_i$  with a 1-sparse vector $v_i \in \{0, 1\}^{2^k}$, where each entry of $v_i$ corresponds to one of the $2^{k}$ possible values of block $B_i$.
That is, if block $B_i$ of $x$ equals the binary representation of the number $j \in [2^k]$, then Alice replaces block $B_i$ with the vector $e_j \in \{0, 1\}^{2^k}$, where $e_j$ denotes the vector with a 1 in coordinate $j$ and 0s elsewhere.

Thus, Alice now has an $n'/k=m$-sparse derived input $y$ over a universe of size $(n'/k) \cdot 2^k=n$.
 Merlin looks at Bob's input to see what is the index $\iota$ of the dense vector $x$ that Bob is interested in. Merlin then tells Bob the index $\ell$ such that 
 $\ell= 2^k(\iota-1)+j$, where $B_i$ is the block that $\iota$ is located in, and block $B_i$ of Alice's input $x$ equals the binary representation of the number $j \in [2^k]$.
 Notice $\ell$ can be specified with $\log n = O(\sqrt{m})$ bits.

Alice and Bob's now use the assumed AMA-protocol for sparse disjointness to establish whether $y_{\ell}=1$.
If they are convinced of this, then Bob can deduce the value of \emph{all} the bits in block $B_i$
 of the original dense vector $x$, and in particular, the value of $x_{\iota}$.
  
  This yields an MAMA protocol for the dense \idx\ problem on $n'=\Omega(m^{3/2})$ bits. A lower bound of  Klauck and Prakash \cite[Lemma 7]{klauck} 
  implies that the online MAMA complexity of this problem is $\Omega((n')^{1/3}) = \Omega(m^{1/2})$.
   Notice also that the
  total $\hcost$ of our MAMA protocol is $O(\sqrt{m} + \hcost(\cP)) = O(\hcost(\cP))$, while the $\vcost$ is $O(\vcost(\cP))$. 
  Thus, if $\hcost(\cP)=\Omega(\sqrt{m})$, it must be the case that $\vcost$ is $\Omega(\sqrt{m})$ as well. 
  This trivially implies that for any protocol $\cP$ with $\hcost$ \emph{less} than $\sqrt{m}$, $\vcost(\cP)$ must be $\Omega(\sqrt{m})$. 
We conclude the online AMA communication complexity of the problem is $\Omega(m^{1/2})$. This completes the proof.

\end{proof}

\end{document}